\documentclass[pdflatex,sn-mathphys-num]{sn-jnl}% Math and Physical Sciences Numbered Reference Style

\usepackage{graphicx}%
\usepackage{multirow}%
\usepackage{amsmath,amssymb,amsfonts}%
\usepackage{amsthm}%
\usepackage{mathtools}%
\usepackage{bm}%
\usepackage{mathrsfs}%
\usepackage[title]{appendix}%
\usepackage{xcolor}%
\usepackage{textcomp}%
\usepackage{manyfoot}%
\usepackage{booktabs}%
\usepackage{enumitem}

\usepackage{color}

\theoremstyle{plain}
\newtheorem{theorem}{Theorem}
\newtheorem{lemma}{Lemma}
\newtheorem{proposition}{Proposition}
\newtheorem{corollary}{Corollary}
\theoremstyle{remark}
\newtheorem{remark}{Remark}
\theoremstyle{definition}

\newcommand{\M}{\mathcal{M}}
\newcommand{\Ln}{\operatorname{Ln}}

\newcommand{\sigmah}{\sigma_\mathcal{H}}

\newcommand{\rhonh}{\rho_{\mathcal{N}\mathcal{H}}}

\newcommand{\1}{\mathbf{1}}

\newcommand{\RR}{\mathcal{R}}
\newcommand{\B}{\mathcal{B}}
\newcommand{\N}{\mathbb{N}}
\newcommand{\NN}{\mathcal{N}}
\newcommand{\D}{\mathcal{D}}
\newcommand{\eps}{\varepsilon}
\newcommand{\ot}{\otimes}
\newcommand{\norm}[1]{\left\lVert #1\right\rVert}

\newcommand{\Deltaq}{\Delta}

\newcommand{\infs}{\inf_{s\geq 0}}

\begin{document}

\title{Exact exponents for smoothing the max-relative entropy and quantum information decoupling in von Neumann algebras}

\author[1]{ \fnm{Zhiwen} \sur{Lin} }
\email{21b912029@stu.hit.edu.cn}

\author[1]{ \fnm{Hongsen} \sur{Qiu} } % 加 * 号标记为通讯作者
\email{chieughongsen@gmail.com}

\author*[1]{ \fnm{Xinyu} \sur{Zhang} }
\email{xy.zhang@stu.hit.edu.cn}

% --- 定义单位 (编号为1) ---
\affil*[1]{ % 加 * 号标记通讯作者的单位
  \orgdiv{Institute for Advanced Study in Mathematics}, 
  \orgname{Harbin Institute of Technology}, 
  \orgaddress{%
    \city{Harbin}, 
    \postcode{150001}, 
    \country{China}
  }
}

\abstract{
  We establish the exact exponent for smoothing the max-relative entropy
on arbitrary von Neumann algebras. The proof first
develops the required large-deviation and smoothing estimates in the
semifinite setting and then passes to general von Neumann algebras
through Haagerup--Junge--Xu reduction.

  We next study catalytic quantum information decoupling when the
reference system is described by an arbitrary von Neumann algebra. A key ingredient is an operator layer-cake formula valid on arbitrary von Neumann algebras. This formula leads to a convex-split estimate for normal states with a general von Neumann algebraic reference system. Combining this estimate with the smoothing exponent result, we obtain lower and upper bounds on the catalytic-decoupling
reliability function. These bounds coincide below the corresponding critical rate, yielding
the exact decoupling reliability function in that regime.}

\maketitle
\section{Introduction}
Smooth entropies were introduced in the development of one-shot quantum information theory as entropic quantities capable of describing finite-resource information-processing tasks without imposing an independent and identically distributed (i.i.d.) assumption, with the smooth min- and max-entropies playing a central role in quantum cryptography and one-shot information theory \cite{Renner2005,TomamichelColbeckRenner2009,TomamichelColbeckRenner2010}. Within this framework, the max-relative entropy was isolated and systematically studied by Datta as a relative-entropy-type quantity measuring the smallest multiplicative factor by which one positive operator dominates another, and its smoothed version was introduced by optimizing this quantity over states in a prescribed neighbourhood of the original state \cite{Datta2009,Renner2005}. The resulting smooth max-relative entropy provides a common underlying quantity for several one-shot information measures: conditional min-entropy is obtained by minimizing a max-relative entropy over reference states, while max-information is likewise constructed from the max-relative entropy and serves as a one-shot counterpart of mutual information \cite{Datta2009,BertaChristandlRenner2011,CiganovicBeaudryRenner2014,Tomamichel2016}. These quantities have consequently entered the analysis of a wide range of operational problems, including privacy amplification against quantum side information, one-shot channel coding and simulation, quantum state merging, and other protocols whose finite-resource performance is naturally expressed through smooth entropic quantities \cite{Renner2005,TomamichelSchaffnerSmithRenner2011,BertaChristandlRenner2011,DupuisBertaWullschlegerRenner2014}. 

Their connection with conventional asymptotic information theory is provided by asymptotic equipartition properties: for i.i.d. quantum states, suitably smoothed entropies converge to the corresponding von Neumann entropic quantities, thereby recovering the familiar first-order rates from one-shot characterizations \cite{TomamichelColbeckRenner2009,Tomamichel2016}. Beyond the first-order regime, Tomamichel and Hayashi established a hierarchy connecting smooth max-relative entropy, information-spectrum quantities, and hypothesis-testing relative entropy and showed that these quantities exhibit the same leading second-order behavior under appropriate smoothing \cite{TomamichelHayashi2013}; more recently, increasingly tight one-shot relations between smoothed max-relative entropy and hypothesis-testing relative entropy have further clarified the equivalence and differences among the various smoothing formulations \cite{RegulaLamiDatta2026}. At the large-deviation scale, Li, Yao, and Hayashi determined the exact exponential decay rate of the purified-distance modification required to smooth the max-relative entropy in the i.i.d. setting, showing that the exponent is given by a variational formula involving sandwiched Rényi divergences, and used this result to obtain the exact security exponent of quantum privacy amplification in the high-rate regime \cite{LiYaoHayashi2023}; complementary strong-converse analyses subsequently identified sandwiched Rényi divergences of orders below one as the quantities governing the opposite large-deviation regime \cite{LiYao2024StrongConverse}. 

Closely connected with this smooth-entropy framework is ``quantum information decoupling'', whose basic objective is to transform a system correlated with a reference into one that is approximately independent of that reference, a principle that emerged as a unifying technique in the development of quantum state merging, the mother protocol, and quantum channel coding \cite{HorodeckiOppenheimWinter2007,AbeyesingheDevetakHaydenWinter2009,HaydenHorodeckiWinterYard2008}. The decoupling viewpoint proved particularly powerful because successful decoupling from an environment can be converted into operational statements about quantum communication and state transfer, allowing a number of quantum Shannon-theoretic coding results to be derived within a common framework \cite{AbeyesingheDevetakHaydenWinter2009,HaydenHorodeckiWinterYard2008,Dupuis2010}. A general one-shot decoupling theorem was established by Dupuis, Berta, Wullschleger, and Renner, who expressed decoupling conditions directly in terms of smooth entropies and demonstrated their application to one-shot state merging \cite{DupuisBertaWullschlegerRenner2014}, while subsequent work investigated more structured and efficient randomizing operations, including random quantum circuits and partial-decoupling schemes \cite{BrownFawzi2015,WakakuwaNakata2021}. Majenz, Berta, Dupuis, Renner, and Christandl later introduced ``catalytic decoupling'', in which an initially uncorrelated ancillary system is supplied as a catalyst, and obtained a tight one-shot characterization in terms of smooth max-information, thereby placing correlation erasure and state merging within a common decoupling framework \cite{MajenzBertaDupuisRennerChristandl2017}. 

The asymptotic theory has since been refined beyond first-order rates: tight second-order and moderate-deviation characterizations have been obtained for standard decoupling \cite{ShenGaoCheng2023}, whereas Li and Yao initiated a systematic reliability-function analysis of catalytic decoupling and showed that, below a critical cost rate, the exact exponential decay of the decoupling error is characterized by a Legendre-type transform of the sandwiched Rényi mutual information, with general upper and lower bounds outside this regime \cite{LiYao2024}. More recent developments have pursued sharper Rényi-based one-shot and exponential bounds, including approaches that bypass explicit smoothing, improved convex-splitting and decoupling estimates, and any-shot formulations of decoupling, further reinforcing the role of Rényi information quantities in quantitative decoupling theory \cite{Dupuis2023,ChengGaoHircheHuangLiu2025,BertaChengYao2026}.

Von Neumann algebras provide the natural operator-algebraic framework for quantum systems
beyond finite-dimensional matrix algebras.
This point of view has already become a useful language in quantum
information theory.  R\'enyi relative entropies and sandwiched R\'enyi
quantities have been formulated in this setting, e.g., \cite{Jencova2018,HiaiJencova2024}.

Recent work has shown that asymptotic quantum information theory also admits a genuinely
von Neumann algebraic formulation.  Fawzi, Gao and Rahaman established
asymptotic equipartition theorems for smooth entropies on von Neumann
algebras \cite{Gao2025}.  Besides,
Junge and Laracuente studied strong converse exponents for asymptotic hypothesis testing in type-III von Neumann algebras \cite{Junge2025}.  These results indicate that operational asymptotic formulas
are not restricted to finite-dimensional systems.

The first goal of the present paper is to give the exact exponent for smoothing the max-relative entropy.  Let \(\M\) be a von Neumann algebra. Let
\(\rho \in \mathcal{S}(\M)\) and \(\sigma\in \M_*^+\). We prove that, 
\[
    \lim_{n\to\infty}
    -\frac1n
    \log
    \Delta(\rho_n\Vert\sigma_n,nr)
    =
    \frac12
    \sup_{s\ge0}
    s\bigl(r-D_{1+s}(\rho\Vert\sigma)\bigr).
\]
In finite-dimensional cases, Li, Yao, and Hayashi identified this exponent in terms of sandwiched Rényi
divergences \cite{LiYaoHayashi2023}. However, the standard finite-dimensional arguments based on pinching, eigenvalue counting and matrix approximation, cannot survive on arbitrary von Neumann algebras. Our  proof replaces dimension-dependent ingredients by operator-algebraic ones.

We next apply this framework to catalytic quantum information decoupling.  Let \(
    \rho_{\M\mathcal H}\) be a bipartite state where the reference system $\M$ is an arbitrary von Neumann algebra. If
\[
    r\le
    \frac12
    \left.
    \frac{d}{ds}
    \left(
        s I_{1+s}(\M:\mathcal H)_\rho
    \right)
    \right|_{s=1},
\]
then the reliability function of catalytic quantum information decoupling is
\[
    E^{\mathrm{dec}}_{\M:\mathcal H}
    (
        \rho_{\M\mathcal H},r
    )
    =
    \max_{0\le s\le1}
    s\left(
        r-\frac12 I_{1+s}(\M:\mathcal H)_\rho
    \right).
\]
In finite-dimensional cases, Li and Yao characterized the reliability function in terms of the sandwiched R\'enyi mutual
information \cite{LiYao2024}. Our decoupling result shows that the same reliability formula survives in this broader setting.
In this sense, the decoupling exponent is governed by the von Neumann
algebraic structure of the bipartite normal state.

The remainder of this paper is organized as follows. In Sect.~\ref{pre} we introduce the information-theoretic quantities on von Neumann algebras. In Sect.~\ref{exponent} we give  the exact exponent for smoothing the max-relative entropy. In Sect. ~\ref{decouple} we characterize the reliability function of
quantum information decoupling when the reference system is an arbitrary von Neumann algebra.

\section{Preliminaries}\label{pre}

\subsection{Von Neumann algebras and noncommutative \(L^p\)-spaces}
Let \(\mathcal M\subseteq \mathcal B(\mathcal H)\) be a von Neumann
algebra, namely, a unital \(^*\)-subalgebra of
\(\mathcal B(\mathcal H)\) which is closed in the weak operator topology.
Equivalently, by the bicommutant theorem,
\[
    \mathcal M=\mathcal M'',
\]
where
\[
    \mathcal M'
    =
    \{x\in\mathcal B(\mathcal H):xy=yx
      \text{ for all }y\in\mathcal M\}
\]
is the commutant of \(\mathcal M\). We write
\(\mathcal M_+\) for the positive cone of \(\mathcal M\),
\(\operatorname{Proj}(\mathcal M)\) for its projection lattice, and
\(\1\) for its identity; see, for instance,
\cite[Chapters I--III]{TakesakiI}.

Every von Neumann algebra admits a unique isometric predual, denoted by
\(\mathcal M_*\). The predual \(\mathcal M_*\) is canonically
identified with the Banach space of all normal linear functionals on
\(\mathcal M\).
We denote by
\(
    \mathcal M_*^+
\)
the cone of normal positive functionals. For
\(\omega\in\mathcal M_*^+\), one has
\[
    \|\omega\|_{\M_*}=\omega(\1).
\]

For \(\omega\in\mathcal M_*^+\), its support projection
\(s(\omega)\in  \operatorname{Proj}(\mathcal M)\) is the smallest projection \(e\)
such that
\[
    \omega(e)=\omega(\1).
\]
Equivalently,
\[
    \1-s(\omega)
    =
    \sup\{e\in \operatorname{Proj}(\mathcal M):\omega(e)=0\}.
\]
The functional \(\omega\) is faithful if and only if
\(s(\omega)=\1\). A von Neumann algebra \(\mathcal M\) is called \(\sigma\)-finite if it
admits a normal faithful  state.

For von Neumann algebras \(\mathcal M\) and \(\mathcal N\), their von Neumann algebra tensor product is denoted by
\[
    \mathcal M\bar{\otimes}\mathcal N.
\]
If \(\omega\in\mathcal M_*^+\) and
\(\eta\in\mathcal N_*^+\), then
\(\omega\otimes\eta\) denotes the corresponding normal product functional
on \(\mathcal M\bar{\otimes}\mathcal N\).

We next recall the tracial construction of noncommutative
\(L^p\)-spaces. A von Neumann algebra is called semifinite if it admits
a  normal semifinite faithful (n.s.f.) trace. 

% Every type-I or type-II von
% Neumann algebra is semifinite, whereas a type-III von Neumann algebra
% admits no normal semifinite faithful (n.s.f.) trace.

Let \((\mathcal M,\tau)\) be a semifinite von Neumann
algebra equipped with an n.s.f. trace $\tau$. A closed densely
defined operator \(x\) on \(\mathcal H\) is affiliated with
\(\mathcal M\) if it commutes with every unitary in \(\mathcal M'\). 
An affiliated operator \(x\) is called \(\tau\)-measurable if for every
\(\varepsilon>0\), there exists a projection \(e\in \operatorname{Proj}(\mathcal M)\) such that
\[
    e\mathcal H\subseteq\operatorname{Dom}(x),
    \quad
    \tau(\1-e)<\varepsilon.
\]
The collection of all \(\tau\)-measurable operators affiliated with
\(\mathcal M\) forms a \(^*\)-algebra, denoted by
\(L^0(\mathcal M,\tau)\); see \cite{Nelson1974,PX03,Xu07}.

For \(1\leq p<\infty\), the tracial noncommutative \(L^p\)-space is
\[
    L^p(\mathcal M,\tau)
    =
    \{x\in L^0(\mathcal M,\tau):\tau(|x|^p)<\infty\},
\]
equipped with the norm
\(
    \|x\|_p
    =
    \tau(|x|^p)^{1/p}.
\)
We also put
\[
    L^\infty(\mathcal M,\tau)=\mathcal M
\]
with its usual operator norm. 
In particular,
\[
    L^1(\mathcal M,\tau)\cong\mathcal M_*.
\]
Whenever a normal semifinite faithful trace is fixed, we identify normal
positive functionals with their \(L^1\)-densities.

We denote by
\[
    \operatorname{Proj}_\tau(\mathcal M)
    =
    \{e\in \operatorname{Proj}(\mathcal M):\tau(e)<\infty\}
\]
the set of finite-trace projections.

\begin{remark}
The above construction contains the usual finite- and infinite-dimensional quantum systems as special cases. If
\[
    \mathcal M=M_n(\mathbb C),
    \qquad
    \tau=\operatorname{Tr}_n,
\]
then \(L^p(\mathcal M,\tau)\) is the finite-dimensional Schatten class \(S^p_n\). If \(\mathcal{H}\) is a separable infinite-dimensional Hilbert space and
\[
    \mathcal M=\mathcal{B}(\mathcal H),
    \qquad
    \tau=\operatorname{Tr},
\]
then \(L^p(\mathcal M,\tau)\) is the Schatten class \(S^p(\mathcal{H})\).
\end{remark}

For general von Neumann algebras, we briefly introduce some concepts and adopt standard notation from the theory of operator algebras. We refer to \cite{Kadison1986,PT73,Stratila1981,Takesaki1970,Takesaki1972,Takesaki1973}
 for modular theory and \cite{Haa75,TakII03} for the standard form.

We now fix a standard form
\(
\bigl(\mathcal M,\mathcal{H},
J,\mathcal P\bigr)
\)
of \(\mathcal M\), which consists of an injective $*$-homomorphism $\pi:\M\to \mathcal{B}(\mathcal{H})$, an anti-linear isometry $J$ on $\mathcal{H}$, and a self-dual cone $\mathcal{P}$. We identify $\pi(x)=x\in \M$ with its faithful normal representation on $\mathcal{H}$. Such standard form satisfies:
\begin{itemize}[leftmargin=*]
\item[1)] $ J^2=\1,\quad J\M J=\M^\prime,$
\item[2)]$Jx J=x^*,\quad x\in \M\cap \M^\prime,$
\item[3)] $J\xi=\xi,\quad \xi\in \mathcal{P},$
\item[4)] $ xJx\mathcal{P}=\mathcal{P},\quad x\in \M.$
\end{itemize} 
Every \(\omega\in\mathcal M_*^+\) admits a unique vector
representative \(\xi_\omega\in\mathcal P\) satisfying
\[
\omega(x)
=
\langle\xi_\omega,x\xi_\omega\rangle,
\qquad x\in\mathcal M.
\]
In this sense, the support projection $s(\omega)$ is then the projection onto $\overline{\M^\prime \xi_\omega}$.
% and
% fix a faithful normal state \(\varphi\in\mathcal M_*^+\).
% Let \(\xi_\varphi\in\mathcal P_{\mathcal M}\) be the standard-form
% representative of \(\varphi\), so that
% \[
%     \varphi(x)
%     =
%     \langle\xi_\varphi,x\xi_\varphi\rangle,
%     \qquad x\in\mathcal M.
% \]
% Since \(\varphi\) is faithful, \(\xi_\varphi\) is cyclic and
% separating for \(\mathcal M\). 
For $\psi,\varphi\in \M_*^+$, define \(S_{\psi,\varphi}\) as the closure of 
\[
    S_{\psi,\varphi}(x\xi_\varphi+\eta)=s(\varphi)x^*\xi_\psi,
    \qquad x\in\mathcal M,\quad\eta\in \M \xi_{\varphi}^{\perp}
\]
and put
\(
    \Delta_{\psi,\varphi}= S_{\psi,\varphi}^* S_{\psi,\varphi}
\). In particular, the modular operator is given by $\Delta_{\varphi}=S_{\varphi}^*S_{\varphi}$ with $S_{\varphi}=S_{\varphi,
\varphi}$. 

The modular automorphism group associated with \(\varphi\) is
given by
\[
    \sigma_t^\varphi(x)
    =
    \Delta_\varphi^{it}x\Delta_\varphi^{-it},
    \qquad x\in\mathcal M,\quad t\in\mathbb R.
\]
It is a $\sigma$-weakly continuous one-parameter group of
$*$-automorphisms satisfying
\[
    \varphi\circ\sigma_t^\varphi=\varphi,
    \qquad t\in\mathbb R.
\]
If $\varphi$ is tracial, then $\Delta_\varphi=1$ and
$\sigma_t^\varphi=\operatorname{id}_{\mathcal M}$ for every
$t\in\mathbb R$.

\subsection{Information-theoretic concepts on von Neumann algebras}
Several works have formulated the framework of quantum
information theory in the setting of von
Neumann algebras, e.g., see \cite{Hiaibook,HiaiJencova2024,Gao2025}. We briefly review the
relevant concepts here, with the notation and definitions given in the previous part of this section. 

We define the normalized positive states as
\[
    \mathcal S(\mathcal M)
    :=
    \bigl\{
        \omega\in\mathcal M_*^+:
        \omega(\1)= 1
    \bigr\},
\]
and the subnormalized positive states as
\[
    \mathcal S_{\leq}(\mathcal M)
    :=
    \bigl\{
        \omega\in\mathcal M_*^+:
        \omega(\1)\leq 1
    \bigr\}.
\]

% For $\omega\in\mathcal M_*^+$, let
% $h_\omega\in L^1(\mathcal M)_+$ denote the Haagerup $L^1$-density of
% $\omega$.  Thus
% \[
%     \omega(x)=\operatorname{tr}(xh_\omega),
%     \qquad x\in\mathcal M,
% \]
% where $\operatorname{tr}$ is the canonical functional on
% $L^1(\mathcal M)$.  

 Next, we introduce the formulation of the following information-theoretic quantities through the standard form and modular operator theory as \cite{Hiaibook}. For $\alpha>1$, the sandwiched R\'enyi divergence of
$\rho,\sigma\in\mathcal M_*^+$ is defined as
\[
    Q_\alpha(\rho\Vert\sigma)= \sup_{\omega\in \mathcal{S}(\M)}\langle
        \xi_\rho,
        \Delta_{\omega,\sigma}^{1/\alpha^\prime}\,\xi_\rho
     \rangle^\alpha,
\]
where $\frac{1}{\alpha}+\frac{1}{\alpha^\prime}=1$. It then follows that,
\[
    D_\alpha(\rho\Vert\sigma)
    =
    \frac{1}{\alpha-1}\log Q_\alpha(\rho\Vert\sigma).
\]  
Throughout this work, $\log$ denotes the logarithm to base 2, whereas $\ln$ denotes the
natural logarithm.

The max-relative entropy is defined as
$$D_\infty(\rho\Vert \sigma)=D_{\max}(\rho\Vert \sigma)=\log \inf \{\lambda> 0\; :\; \rho\leq \lambda \sigma\},$$
% Using the standard form fixed above, for every
% \(\omega\in\mathcal M_*^+\), let
% \(\xi_\omega\in\mathcal P_{\mathcal M}\) denote its unique
% standard-form representative, so that
% \[
%     \omega(x)
%     =
%     \langle\xi_\omega,x\xi_\omega\rangle,
%     \qquad x\in\mathcal M.
% \]
% For $\rho,\sigma\in\mathcal M_*^+$, let
% $\Delta_{\sigma,\rho}$ denote the relative modular operator associated
% with the ordered pair $(\sigma,\rho)$ in the sense of Araki
% \cite{Araki1976-1ent,Araki1977-2ent}. When $\rho$ is faithful,
% $\Delta_{\sigma,\rho}=S_{\sigma,\rho}^*S_{\sigma,\rho}$, where
% $S_{\sigma,\rho}$ is the closure of the densely defined antilinear
% operator
% \[
%     S_{\sigma,\rho}(x\xi_\rho)
%     =
%     x^*\xi_\sigma,
%     \qquad x\in\mathcal M.
% \]
% For general normal positive functionals, the relative modular operator
% is understood in its standard support-reduced form. Notice that the
% ordinary modular operator introduced above is the diagonal special
% case
% \[
%     \Delta_{\varphi,\varphi}=\Delta_\varphi.
% \]
and the Araki relative entropy (see \cite{Araki1976-1ent,Araki1977-2ent}) is defined by
\[
    D_1(\rho\Vert\sigma)
    =
    D(\rho\Vert\sigma)
    =
    \langle
        \xi_\rho,
        \log \Delta_{\rho,\sigma}\,\xi_\rho
     \rangle.
\]

 For $\rho,\sigma\in\mathcal S_{\leq}(\mathcal M)$, we use the
generalized fidelity
\[
    F(\rho,\sigma)
    =\inf_{\omega\in \mathcal{S}(\M)}\langle
        \xi_\rho,
        \Delta_{\omega,\sigma}^{-1}\,\xi_\rho
     \rangle
\]
and the purified distance
\[
    P(\rho,\sigma)
    =
    \sqrt{1-F(\rho,\sigma)^2}.
\]
 For $\rho \in\mathcal S(\mathcal M)$ and $\sigma \in \M_*^+$, we also define the smoothing quantity as
\begin{equation}
\label{eq:general-smoothing-quantity}
    \Delta(\rho\Vert\sigma,\lambda)
    =
    \inf
    \left\{
        P(\rho,\widetilde\rho):
        \widetilde\rho\in\mathcal S_{\leq}(\mathcal M),
        \ \widetilde\rho\leq 2^\lambda\sigma
    \right\},\quad \lambda \in \mathbb{R}.
\end{equation}
\begin{remark}
Equivalent formulation of these quantities through Haagerup $L^p$-spaces is available. Readers can refer to \cite{Haagerup1979,Ter81} for Haagerup 
\(L^p\)-spaces and the formulation could be found in \cite{Jencova2018}.
\end{remark}

Let \(\M\) and \(\mathcal{N}\) be von Neumann algebras. A quantum channel, in the Heisenberg picture, is a normal unital completely
positive (UCP) map
\[
\Phi:\mathcal{N}\longrightarrow \M.
\]
Its preadjoint, in the so-called Schrödinger picture, is the map
\[
\Phi_*:\M_*\longrightarrow \mathcal{N}_*,
\qquad
\Phi_*(\omega):=\omega\circ\Phi.
\]
Since \(\Phi\) is normal and UCP, 
\(\Phi_*\) maps normal states to normal states and normal
subnormalized positive functionals to normal subnormalized positive
functionals.

We shall repeatedly use the standard data-processing properties. For every
\(\alpha> 1\),
\[
Q_\alpha\bigl(\Phi_*(\rho)\Vert\Phi_*(\sigma)\bigr)
\leq
Q_\alpha(\rho\Vert\sigma),\quad
D_\alpha\bigl(\Phi_*(\rho)\Vert\Phi_*(\sigma)\bigr)
\leq
D_\alpha(\rho\Vert\sigma).
\]
Besides, for the Araki relative entropy $D_1$ and the max-relative entropy $D_\infty$, the data-processing property also holds.

For \(\rho,\sigma\in \mathcal{S}_{\leq}(\M)\), the generalized fidelity and the
purified distance satisfy
\[
F\bigl(\Phi_*(\rho),\Phi_*(\sigma)\bigr)
\geq
F(\rho,\sigma),\quad
P\bigl(\Phi_*(\rho),\Phi_*(\sigma)\bigr)
\leq
P(\rho,\sigma).
\]
Consequently, for $\rho\in\mathcal{S}(\M)$ and every
\(\sigma\neq 0\) and every \(\lambda\in\mathbb R\),
\[
\Delta\bigl(
\Phi_*(\rho)\Vert\Phi_*(\sigma),\lambda
\bigr)
\leq
\Delta(\rho\Vert\sigma,\lambda).
\]

The data-processing inequalities for the Araki relative entropy, max-relative entropy and
the sandwiched R\'enyi divergence on arbitrary von Neumann algebras
are standard; see, e.g.,
\cite{Jencova2018,BertaScholzTomamichel2018,HiaiJencova2024}.
The fidelity inequality is the monotonicity of the transition
probability under normal quantum channels, and the purified-distance
inequality follows from its definition; see
\cite{Uhlmann1976,BertaFurrerScholz2016}.

% For the max-relative entropy, if \(
% \rho\leq c\sigma\),
% then positivity of \(\Phi_*\) gives
% \[
% \Phi_*(\rho)\leq c\,\Phi_*(\sigma).
% \]
% Taking the infimum over \(c\) proves the assertion.

Finally, let \(\widetilde\rho\in \mathcal{S}_{\leq}(\M)\) satisfy
\(
\widetilde\rho\leq 2^\lambda\sigma
\).
Then
\[
\Phi_*(\widetilde\rho)
\leq
2^\lambda\Phi_*(\sigma),
\]
and data processing of the purified distance gives
\[
P\bigl(
\Phi_*(\rho),\Phi_*(\widetilde\rho)
\bigr)
\leq
P(\rho,\widetilde\rho).
\]
Taking the infimum over all admissible \(\widetilde\rho\) proves the claim for the smoothing quantity.

When $\M$ is semifinite, the previous definition naturally recovers
\begin{align*}
\mathcal{S}(\M,\tau)&=\{\rho\in L^1(\M,\tau)_+:\tau(\rho)=1\}\\
\mathcal{S}_{\leq}(\M,\tau)&=\{\rho\in L^1(\M,\tau)_+:\tau(\rho)\leq 1\}
\end{align*}
for the set of normalized and subnormalized states.

Moreover, for \(\rho,\sigma\in \mathcal{S}_{\leq}(\M,\tau)\), one also has
\[
F(\rho,\sigma)
=
\tau\!\left(\left|\rho^{1/2}\sigma^{1/2}\right|\right)
+
\sqrt{(1-\tau(\rho))(1-\tau(\sigma))},
\]
and
\[
P(\rho,\sigma)
=
\sqrt{1-F(\rho,\sigma)^2}.
\]
For \(\lambda\in\mathbb R\), the smoothing quantity is 
\[
\Delta(\rho\Vert\sigma,\lambda)
=
\inf\Bigl\{
P(\rho,\widetilde\rho):
\widetilde\rho\in \mathcal{S}_{\leq}(\M),\
\widetilde\rho\leq 2^\lambda\sigma
\Bigr\}.
\]

% In
% particular, treating $\rho,\sigma$ as elements in $\M_*$, we have
% \[
% F(\rho,\sigma)^2=
% \sup_{\pi,\xi,\eta}
% \left|\langle \xi,\eta\rangle\right|^2,
% \]
% where the supremum is taken over all normal representations
% \(\pi:\M\to \mathcal{B}(\mathcal{H})\)  and all representatives
% \(\xi,\eta\in \mathcal{H}\) satisfying
% \[
% \rho(x)=\langle \xi,\pi(x)\xi\rangle,
% \qquad
% \sigma(x)=\langle \eta,\pi(x)\eta\rangle,
% \qquad x\in \M .\]
% For details, see \cite{Uhlmann1976}. Consequently, we shall use standard properties of the fidelity and the purified
% distance for states on von Neumann algebras, such as data processing inequality, monotonicity and other related properties, with the Stinespring representation of the channel.  Hence the Fuchs--van de Graaf inequality still holds in this setting.

For non-zero \(\rho,\sigma\in L^1(\M,\tau)_+\) and \(\alpha> 1\), we have
\[
Q_\alpha(\rho\Vert\sigma)
=
\tau\!\left[
\left(
\sigma^{\frac{1-\alpha}{2\alpha}}
\rho
\sigma^{\frac{1-\alpha}{2\alpha}}
\right)^\alpha
\right],\quad D_\alpha(\rho\Vert\sigma)=\frac{1}{\alpha-1}\log Q_\alpha(\rho\Vert\sigma).
\]
Besides, Araki relative entropy coincides with the Umegaki relative entropy:
\[
D_1(\rho\Vert\sigma)
=
\tau (\rho (\log \rho-\log \sigma)).
\]

When $\M$ and $\mathcal{N}$ are tracial, the Schrödinger picture channel $\Phi_*$ is indeed a normal, completely positive and trace-preserving map (CPTP) from $L^1(\M,\tau_\M)$ to $L^1(\NN,\tau_{\NN})$. For this case, we can also use the data-processing inequality of the mentioned quantities, e.g., see \cite{Jencova2018,HiaiJencova2024}.

For tensor powers, we write
\[
\M_n=\M^{\bar\otimes n}.
\]
If \(\rho,\sigma\in \M^+_*\), then
\[
\rho_n=\rho^{\otimes n},
\quad
\sigma_n=\sigma^{\otimes n}
\]
are regarded as elements of \((\M_n)^+_*\).

For $\rho\in \mathcal{S}(\M)$ and $\sigma \in \M_*^+$ on arbitrary von Neumann algebras,
the sandwiched R\'enyi divergence satisfies the following standard
properties for \(\alpha> 1\):
\begin{enumerate}
\item \noindent tensor-product additivity,
\[
D_\alpha(\rho_1\otimes\rho_2\Vert\sigma_1\otimes\sigma_2)
=
D_\alpha(\rho_1\Vert\sigma_1)
+
D_\alpha(\rho_2\Vert\sigma_2);
\]

\item monotonicity,
\[
1<\alpha\leq\beta
\quad\Longrightarrow\quad
D_\alpha(\rho\Vert\sigma)
\leq
D_\beta(\rho\Vert\sigma);
\]

\item the limiting relations
\[
\lim_{\alpha\downarrow1}
D_\alpha(\rho\Vert\sigma)
=
D(\rho\Vert\sigma),\quad
\lim_{\alpha\uparrow\infty}
D_\alpha(\rho\Vert\sigma)
=
D_{\max}(\rho\Vert\sigma);
\]

\item joint lower semicontinuity on
\(\M_*^+\times \M_*^+\) with respect to the predual norm.
\end{enumerate}
These properties are understood in the extended-valued sense whenever $\infty$ appears.

We finish this subsection with the following Legendre transform lemma related to the relative entropy, which will be repeatedly used throughout the paper. We include the proof in Appendix~\ref{app-ltrans}.
\begin{lemma}\label{l-trans}
Let \(I\) be a directed set. Let
\[
0\neq\rho_i,\sigma_i\in(\mathcal M_i)_*^+,
\qquad
0\neq\rho,\sigma\in\mathcal M_*^+.
\]
Assume that
\(
\rho_i(\1)\to \rho(\1),\,\sigma_i(\1)\to \sigma(\1)
\).
We also put
\[
\varphi_i(s)
=
sD_{1+s}(\rho_i\Vert\sigma_i),
\quad
\varphi(s)
=
sD_{1+s}(\rho\Vert\sigma),\quad s>0.
\] 
with $\varphi_i(0)=\log \rho_i(\1)$, $\varphi(0)=\log \rho(\1)$.

If \(
a<D_\infty(\rho\Vert\sigma)
\) and $\varphi_i(s)\to \varphi(s)$ pointwise, then
\[
J_i(a)=\inf_{s\geq0}
\{\varphi_i(s)-a s\}
\longrightarrow
\inf_{s\geq0}
\{\varphi(s)-as\}=J(a).
\]
Equivalently,
\[
E_i(a)=\sup_{s\geq0}
\{a s-\varphi_i(s)\}
\longrightarrow
\sup_{s\geq0}
\{as-\varphi(s)\}=E(a).
\]
\end{lemma}

\subsection{Haagerup--Junge--Xu reduction}
\label{reduct}

In this part, we record the form of the Haagerup--Junge--Xu reduction method.

For information-theoretic questions involving finitely
many objects, it is reasonable to admit a normal faithful state, hence we can assume the underlying von Neumann algebra is $\sigma$-finite. Let $\mathcal M$ be a von Neumann algebra with a normal faithful
state $\varphi\in\mathcal M_*$.  Set
\[
    G=\bigcup_{m\geq 1}2^{-m}\mathbb Z,
\]
where $G$ is equipped with the discrete topology. Let $\lambda(G)$ be the left regular representation of $G$. Define the crossed product as
\begin{equation*}
    \mathcal R :=
    \mathcal M\rtimes_{\sigma^\varphi}G=\{\pi(\M),\lambda(G)\}^{\prime\prime}\subset \M\bar{\otimes}\mathcal{B}(\ell_2(G))
\end{equation*}
with the following embeddings:
\begin{align*}
\iota&:\M\to \mathcal{R},\quad \iota(x)=\sum_{g\in G}\sigma_{g^{-1}}^{\varphi}(x)\otimes |g\rangle\langle g|,\\
\lambda &: G\to \mathcal{R},\quad \lambda(g)(a\otimes b)=a\otimes \lambda(g)b,\quad a\in \mathcal{H},\quad b\in \ell_2(G).
\end{align*}
Here $\mathcal{H}$ is the Hilbert space in the standard form $(\M, \mathcal{H},J,\mathcal{P})$.  We identify
$\mathcal M$ with $\iota(\mathcal M)$ and view $\M\subset \mathcal{R}$ as a subalgebra. Then there is a normal faithful
conditional expectation
\(
    E:\mathcal R\longrightarrow\mathcal M
\)
determined on the algebraic crossed product by
\begin{equation}
\label{eq:canonical-expectation}
    E\bigl(\lambda(g)\iota(x)\bigr)
    =
    \begin{cases}
        x, & g=0,\\
        0, & g\neq 0.
    \end{cases}
\end{equation}
The functional
\(
    \widehat\varphi=\varphi\circ E
\) is then a faithful normal state on $\mathcal R$.

We shall use the following form of Haagerup--Junge--Xu reduction theorem
\cite[Theorem~2.1 and Lemma~2.7]{HaagerupJungeXu2010}.

\begin{theorem}[Haagerup--Junge--Xu reduction]
\label{thm:haagerup-reduction}
There exists an increasing sequence
$(\mathcal R_m)_{m\geq 1}$ of finite von Neumann subalgebras of
$\mathcal R$ and normal faithful conditional expectations
\[
    \Phi_m:\mathcal R\longrightarrow\mathcal R_m
\]
such that
\[
    \widehat\varphi\circ\Phi_m=\widehat\varphi,
    \qquad
    \sigma_t^{\widehat\varphi}\circ\Phi_m
    =
    \Phi_m\circ\sigma_t^{\widehat\varphi},
    \qquad t\in\mathbb R,
\]
and
\[
    \left(
        \bigcup_{m\geq 1}\mathcal R_m
    \right)''
    =
    \mathcal R.
\]
Moreover,
\[
    \Phi_m(x)\longrightarrow x
\]
in the $\sigma$-strong topology for every $x\in\mathcal R$.
Each $\mathcal R_m$ is $\sigma$-finite and finite, and therefore
admits a normal faithful finite trace.
\end{theorem}

The maps $E$ and $\iota$ form a normal completely positive retract:
\[
    E\circ\iota=\operatorname{id}_{\mathcal M}.
\]
This observation allows us to lift information-theoretic quantities
from $\mathcal M$ to $\mathcal R$ and then recover them without loss.

% \begin{lemma}
% \label{lem:ucp-retract}
% Let $\mathcal A,\mathcal B$ be von Neumann algebras.
% Suppose that
% \[
% \iota:\mathcal A\to\mathcal B,
% \qquad
% E:\mathcal B\to\mathcal A
% \]
% are normal unital completely positive maps satisfying
% \[
% E\circ\iota=\operatorname{id}_{\mathcal A}.
% \]
% Then all quantities appearing in this paper which satisfy
% data processing in both directions are invariant under this retract.
% In particular,
% \[
% D
% (\rho\circ E\Vert\sigma\circ E)
% =
% D
% (\rho\Vert\sigma),
% \]
% and similarly for $D_{\max}$, purified distance, and the smoothing
% quantity $\Delta$.
% \end{lemma}

\begin{proposition}
\label{prop:crossed-product-retract}
For $\omega\in\mathcal M_*^+$, define its canonical lift by
\[
    \widehat\omega=\omega\circ E\in\mathcal R_*^+.
\]
Then, for every $\rho,\sigma\in\mathcal M_*^+$ and every
$\alpha>1$,
\begin{equation*}
    Q_\alpha
    (\widehat\rho\Vert\widehat\sigma)
    =
    Q_\alpha
    (\rho\Vert\sigma),\quad D_\alpha
    (\widehat\rho\Vert\widehat\sigma)
    =
    D_\alpha
    (\rho\Vert\sigma),\quad D_1
    (\widehat\rho\Vert\widehat\sigma)
    =
    D_1
    (\rho\Vert\sigma).
\end{equation*}
For $\rho,\sigma\in\mathcal S_{\leq}(\mathcal M)$,
\begin{align*}
    P(\widehat\rho,\widehat\sigma)
    =
    P(\rho,\sigma).
\end{align*}
For $\rho\in\mathcal S(\mathcal M)$, $\sigma\in \M_*^+$ and $\lambda\in \mathbb{R}$,
$$   \Delta
    (\widehat\rho\Vert\widehat\sigma,\lambda)
    =
    \Delta
    (\rho\Vert\sigma,\lambda).$$
\end{proposition}

\begin{proof}
Applying data processing for the sandwiched R\'enyi divergence to
$E:\mathcal R\to\mathcal M$ gives
\[
    Q_\alpha
    (\widehat\rho\Vert\widehat\sigma)
    \leq
    Q_\alpha
    (\rho\Vert\sigma).
\]
Applying data processing to
$\iota:\mathcal M\to\mathcal R$ and using
$\widehat\rho\circ\iota=\rho$ and
$\widehat\sigma\circ\iota=\sigma$ gives the reverse inequality. The results for relative entropy and purified distance follow by applying the same two-sided argument.

% For \eqref{eq:positive-part-retract}, note that
% \[
% \begin{aligned}
%     \bigl\|
%         (\widehat\rho-c\widehat\sigma)_+
%     \bigr\|
%     &=
%     \sup_{0\leq X\leq 1_{\mathcal R}}
%     (\rho-c\sigma)(E(X))\\
%     &\leq
%     \sup_{0\leq x\leq 1_{\mathcal M}}
%     (\rho-c\sigma)(x).
% \end{aligned}
% \]
% The reverse inequality follows by taking $X=\iota(x)$.

To prove the invariance of the smoothing quantity, first let
$\widetilde\rho\in\mathcal S_{\leq}(\mathcal M)$ satisfy
$\widetilde\rho\leq 2^\lambda\sigma$.  Then
\[
    \widehat{\widetilde\rho}
    \leq
    2^\lambda\widehat\sigma,
\]
and the invariance of the purified distance yields
\[
    \Delta
    (\widehat\rho\Vert\widehat\sigma,\lambda)
    \leq
    \Delta
    (\rho\Vert\sigma,\lambda).
\]
Conversely, if
$\eta\in\mathcal S_{\leq}(\mathcal R)$ satisfies
$\eta\leq 2^\lambda\widehat\sigma$, then
\[
    \eta\circ\iota
    \leq
    2^\lambda\sigma.
\]
By monotonicity of purified distance under restriction,
\[
    P_{\mathcal M}
    (\rho,\eta\circ\iota)
    \leq
    P_{\mathcal R}
    (\widehat\rho,\eta).
\]
Taking the infimum over $\eta$ proves the reverse inequality.
\end{proof}

For the approximation step, we choose the reference
state in the crossed product determined by the second functional $\sigma$. More
precisely, assume that $\sigma$ is faithful, we take
\[
    \varphi=\frac{\sigma}{\sigma(\1)}.
\]
Then
\[
    \widehat\sigma=\sigma\circ E
    =
    \sigma(\1)\widehat\varphi,
\]
and hence
\begin{equation}
\label{eq:sigma-preserved-Phi-m}
    \widehat\sigma\circ\Phi_m
    =
    \widehat\sigma
    \qquad
    \forall m\geq 1.
\end{equation}

Put
\[
    \rho_m:=\widehat\rho|_{\mathcal R_m},
    \qquad
    \sigma_m:=\widehat\sigma|_{\mathcal R_m}.
\]
We shall frequently use the following martingale recovery property for the information theoretic quantities under increasing von Neumann subalgebras.
\begin{proposition}[Martingale recovery]
\label{prop:haagerup-information-approximation}
With the notation above, for every $\alpha \geq 1$,
\begin{equation*}
    D_\alpha
    (\rho_m\Vert\sigma_m)
    \nearrow
    D_\alpha
    (\widehat\rho\Vert\widehat\sigma)
    =
    D_\alpha
    (\rho\Vert\sigma).
\end{equation*}
and for $\alpha>1$,
\begin{equation*}
   Q_\alpha
    (\rho_m\Vert\sigma_m)
    \nearrow
    Q_\alpha
    (\widehat\rho\Vert\widehat\sigma)
    =
    Q_\alpha
    (\rho\Vert\sigma).
\end{equation*}

Besides, for every $\lambda\in\mathbb R$, we also have
\begin{align*}
    \Delta
    (\rho_m\Vert\sigma_m,\lambda)
    \nearrow\Delta_{\mathcal R}
    (\widehat\rho\Vert\widehat\sigma,\lambda)=
    \Delta
    (\rho\Vert\sigma,\lambda).
\end{align*}

\end{proposition}

\begin{proof}
The convergence
for $D_\alpha$ and $Q_\alpha$ with $\alpha>1$ is the martingale convergence
theorem for the sandwiched R\'enyi divergence
\cite[Theorem~3.1]{HiaiMosonyi2023}. For $D_1$, one can refer to \cite[Proposition 2.2]{Gao2025}.

It remains to prove the approximation of the smoothing quantity.  Restriction from
$\mathcal R_{m+1}$ to $\mathcal R_m$ is a normal quantum channel.
Therefore
\[
    \Delta
    (\rho_m\Vert\sigma_m,\lambda)
    \leq
    \Delta
    (\rho_{m+1}\Vert\sigma_{m+1},\lambda)
    \leq
    \Delta
    (\widehat\rho\Vert\widehat\sigma,\lambda).
\]
The state-preserving conditional expectations $\Phi_m$ extend to
contractions on $\mathcal{R}_*$.
The noncommutative martingale convergence theorem therefore gives
\[\Vert\widehat\rho-
         \widehat\rho\circ\Phi_m\Vert_{\RR_*} \to 0
.\]
Put
\(
    \varepsilon_m=
    P_{\mathcal R}
    \bigl(
        \widehat\rho,
        \widehat\rho\circ\Phi_m
    \bigr)
\). By the predual-norm convergence and the continuity of
purified distance,
\[
    \varepsilon_m\longrightarrow 0^+.
\]
Let
$\omega_m\in\mathcal S_{\leq}(\mathcal R_m)$ satisfy
\[
    \omega_m\leq 2^\lambda\sigma_m.
\]
Then its lift
\(
    \omega_m^\uparrow:=\omega_m\circ\Phi_m
\)
satisfies, by \eqref{eq:sigma-preserved-Phi-m},
\[
    \omega_m^\uparrow
    \leq
    2^\lambda\sigma_m\circ\Phi_m
    =
    2^\lambda\widehat\sigma.
\]
Using the triangle inequality and
Proposition~\ref{prop:crossed-product-retract} for the retract
$\Phi_m:\mathcal R\to\mathcal R_m$, we obtain
\[
\begin{aligned}
    \Delta_{\mathcal R}
    (\widehat\rho\Vert\widehat\sigma,\lambda)
    &\leq
    P_{\mathcal R}
    (\widehat\rho,\omega_m^\uparrow)\\
    &\leq
    P_{\mathcal R}
    (
        \widehat\rho,
        \widehat\rho\circ\Phi_m
    )
    +
    P_{\mathcal R}
    (
        \widehat\rho\circ\Phi_m,
        \omega_m\circ\Phi_m
    )\\
    &=
    \varepsilon_m
    +
    P_{\mathcal R_m}(\rho_m,\omega_m).
\end{aligned}
\]
Taking the infimum over $\omega_m$ and then letting $m\to\infty$
proves the desired result.
\end{proof}

% \begin{remark}
% \label{rem:continuous-core-warning}
% The preceding argument uses the discrete crossed product
% \eqref{eq:haagerup-discrete-crossed-product} and the finite
% subalgebras $\mathcal R_m$.  It is not sufficient merely to embed
% $\mathcal M$ into its continuous core.  Indeed, Haagerup's $L^p(\mathcal M)$ is generally different
% from the ordinary tracial space $L^p(\mathcal C,\tau)$ of the
% continuous core $\mathcal C=\mathcal M\rtimes_{\sigma^\varphi}
% \mathbb R$; a nonzero element of $L^p(\mathcal M)$ need not have
% finite $L^p(\mathcal C,\tau)$-norm.  The finite-algebra
% approximation in Theorem~\ref{thm:haagerup-reduction} is precisely
% what permits one to use tracial arguments without making this
% identification.
% \end{remark}

To conclude, the reduction scheme used below is therefore
\[
    \mathcal M
    \xrightarrow{\ \omega\mapsto\widehat{\omega}\circ E\ }
    \mathcal R
    \xrightarrow{\ \text{restriction}\ }
    \mathcal R_m,
\]
followed by the limit $m\to\infty$ and the invariance of information theoretic quantities between $\mathcal R$ and $\mathcal M$.  Since each $\mathcal R_m$ is finite,
all results established for semifinite von Neumann
algebras apply to $\mathcal R_m$.

\section{Smoothing Exponents in von Neumann algebras}\label{exponent}
In this section, we establish the exponent for smoothing the max-relative entropy on von Neumann algebras. By the Haagerup--Junge--Xu reduction method, we first prove the results in the semifinite von Neumann algebra setting,  then we extend the results to general von Neumann algebras.
\subsection{A semifinite Datta--Renner type estimate}
The Datta--Renner type
estimate \cite[Theorem~5]{RegulaLamiDatta2026} for the purified distance gives an upper bound on the exponent in finite-dimensional cases. In this section we prove the corresponding fidelity estimate in
semifinite von Neumann algebra setting.

We first prove the following semifinite gentle measurement estimate.

\begin{proposition}
\label{lem:tracial-gentle}
Let \(\rho\in \mathcal{S}(\mathcal M,\tau)\). Let
\(G\in\mathcal M\) such that \(0\le G\le \mathbf{1}\).  Suppose that
\[
    \tau(\rho G^2)\ge 1-\varepsilon
\]
for some \(\varepsilon\in[0,1]\).  Then
\[
    F(\rho,G\rho G)
    \ge 1-\varepsilon
    \ge (1-\varepsilon)^2 .
\]
Moreover, if \(p=\tau(G\rho G)=\tau(\rho G^2)>0\), then
\[
    F\!\left(\rho,\frac{G\rho G}{p}\right)
    \ge \sqrt{1-\varepsilon}
    \ge 1-\varepsilon .
\]
\end{proposition}

\begin{proof}
Since \(G\ge0\) and $\rho$ is a state, we have
\[
    F(\rho,G\rho G)=\tau\left(\Big|\rho^{1/2}(G\rho G)^{1/2}\Big|\right).\]
Applying $\tau(|ab|)=\tau(|ba|)$ for $a,b$ positive and $\tau$-measurable, we get
\[
     F(\rho,G\rho G)=\tau\left(\Big|(G\rho G)^{1/2}\rho^{1/2}\Big|\right)
    =\tau\!\left((\rho^{1/2}G\rho G\rho^{1/2})^{1/2}\right)
    =
    \tau(\rho^{1/2}G\rho^{1/2})
    =
    \tau(\rho G).
\]
Since \(0\le G\le \mathbf{1}\), we have \(G^2\le G\), hence
\[
    \tau(\rho G)\ge \tau(\rho G^2)\ge 1-\varepsilon .
\]
This proves the first estimate.  For the normalized state, homogeneity
of fidelity in the second argument gives
\[
    F\!\left(\rho,\frac{G\rho G}{p}\right)
    =
    \frac{F(\rho,G\rho G)}{\sqrt p}
    \ge
    \frac{p}{\sqrt p}
    =
    \sqrt p
    \ge
    \sqrt{1-\varepsilon}
    \ge
    1-\varepsilon .
\]
\end{proof}

We next introduce the Kubo-Ando geometric mean in the bounded setting, which will be applied to elements of $L^1(\mathcal{M},\tau)_+\cap \mathcal{M}$.
Recall that for bounded positive invertible elements \(X,Y\in
\mathcal{B}(\mathcal{H})\), the Kubo--Ando geometric mean is
\[
    X\#Y
    :=
    X^{1/2}\left(X^{-1/2}YX^{-1/2}\right)^{1/2}X^{1/2}.
\]
Equivalently, if \(0\le b\le c\in L^1(\mathcal{M},\tau)_+\cap \mathcal{M}\) and \(c\) is invertible, then
\begin{align}\label{inver}
    G=(c^{-1})\# b
    =
    c^{-1/2}\left(c^{1/2}bc^{1/2}\right)^{1/2}c^{-1/2}
\end{align}
is the unique positive solution of
\(
    GcG=b
\).
Moreover, by monotonicity of the geometric mean,
\begin{align}\label{contract}
    0\le G\le (c^{-1})\# c = \mathbf{1} .
\end{align}

\begin{lemma}
\label{lem:bounded-tracial-DR}
Let \(a,b,r\in L^1(\mathcal M,\tau)_+\cap\mathcal M\) satisfy
\[
    \tau(a)=1, \quad a\le b+r ,\quad \tau(r)=\varepsilon<1.
\]
Also assume that $b+r$ is invertible. Then there
exists
\(
    \widetilde a
\)
such that
\[
    \widetilde a\in \mathcal{S}_{\leq}(\mathcal M,\tau),\quad \widetilde a\le b,
\]
and
\[
    F(a,\widetilde a)
    \ge (1-\varepsilon)^2 .
\]
\end{lemma}

\begin{proof}
Put $c=b+r$. Since $c$ is invertible, put
\[
    G=(c^{-1})\# b
    =
    c^{-1/2}\left(c^{1/2}bc^{1/2}\right)^{1/2}c^{-1/2}.
\]
By \eqref{inver} and \eqref{contract}, we get
\[
    0\le G\le \mathbf{1},
    \qquad
    GcG=b.
\]
Set that
\( \widetilde a=GaG\).
Since \(a\le c\), we get
\[
    \widetilde a=GaG\le GcG=b.
\]
In particular \(\tau(\widetilde a)\le\tau(a)\le1\).

It remains to estimate the fidelity.  Since \(\mathbf{1}-G^2\ge0\), the
assumption \(a\le c\) gives
\[
\begin{aligned}
    1-\tau(aG^2)=
    \tau\bigl(a(\mathbf{1}-G^2)\bigr)                                     \le
    \tau\bigl(c(\mathbf{1}-G^2)\bigr)=
    \tau(c)-\tau(cG^2).
\end{aligned}
\]
By traciality, we get
\begin{align*}  1-\tau(aG^2)\leq
    \tau(c)-\tau(GcG) =
    \tau(b+r)-\tau(b)  =
    \tau(r)=\varepsilon .
\end{align*}
Hence we obtain
\[
    \tau(aG^2)\ge1-\varepsilon .
\]
Applying Proposition~\ref{lem:tracial-gentle}, we obtain
\[
    F(a,\widetilde a)
    =
    F(a,GaG)
    \ge
    (1-\varepsilon)^2 .
\]
This proves the claim.
\end{proof}
The following density fact is standard in the theory of \(\tau\)-measurable operators and noncommutative \(L^p\)-spaces; see, for example, Nelson~\cite{Nelson1974} or Fack--Kosaki~\cite{FackKosaki1986}.

\begin{proposition}
Let \((\mathcal M,\tau)\) be a semifinite von Neumann algebra.
Then \(L^1(\mathcal M,\tau)\cap\mathcal M\) is dense in
\(L^1(\mathcal M,\tau)\). Moreover,
\(L^1(\mathcal M,\tau)_+\cap\mathcal M\) is dense in
\(L^1(\mathcal M,\tau)_+\).
\end{proposition}

\begin{proof}
Let \(x\in L^1(\mathcal M,\tau)\) and write \(x=u|x|\) for its polar
decomposition.  Put
\[
        x_n=u|x|\mathbf{1}_{[0,n]}(|x|).
\]
Then \(x_n\in L^1(\mathcal M,\tau)\cap\mathcal M\), and
\[
        \|x-x_n\|_1
        =
        \tau\bigl(|x|\mathbf{1}_{(n,\infty)}(|x|)\bigr)
        \longrightarrow0.
\]
This proves the density of \(L^1(\mathcal M,\tau)\cap\mathcal M\).
If \(x\ge0\), then \(x_n=x\mathbf{1}_{[0,n]}(x)\) is positive, belongs to
\(L^1(\mathcal M,\tau)_+\cap\mathcal M\), and the same computation gives
\[
        \|x-x_n\|_1
        =
        \tau\bigl(x\mathbf{1}_{(n,\infty)}(x)\bigr)
        \longrightarrow0.
\]
\end{proof}

\begin{proposition}\label{trunc}
Let \((\mathcal M,\tau)\) be a semifinite von Neumann algebra. Let
\(0\le b\le c\) with \(b,c\in L^1(\mathcal M,\tau)_+\). For \(0<\delta<N<\infty\), put
\(
        e_{\delta,N}=\mathbf{1}_{[\delta,N]}(c)
\) and define
\[
        b_{\delta,N}=e_{\delta,N}be_{\delta,N},
        \qquad
        c_{\delta,N}=e_{\delta,N}ce_{\delta,N}.
\]
Then
\[
        0\le b_{\delta,N}\le c_{\delta,N},\quad
        b_{\delta,N},c_{\delta,N}
        \in L^1(\mathcal M,\tau)_+\cap\mathcal M.
\]

Moreover, 
\[
        b_{\delta,N}\to b,
        \qquad
        c_{\delta,N}\to c
\]
in \(L^1(\mathcal M,\tau)\) as \(\delta\downarrow0\) and \(N\uparrow\infty\). Finally, inside the corner \(e_{\delta,N}\mathcal M e_{\delta,N}\), one has
\[
        \delta e_{\delta,N}\le c_{\delta,N}\le N e_{\delta,N}.
\]
In particular, \(c_{\delta,N}\) is bounded and invertible in
\(e_{\delta,N}\mathcal M e_{\delta,N}\).

\end{proposition}

\begin{proof}
The inequality \(0\le b_{\delta,N}\le c_{\delta,N}\) follows from \(0\le b\le c\) by compression with the projection \(e_{\delta,N}\).

Since \(e_{\delta,N}\) is a spectral projection of \(c\), it commutes with \(c\). Thus
\[
        c_{\delta,N}
        =
        c\mathbf{1}_{[\delta,N]}(c).
\]
Hence \(
        0\le c_{\delta,N}\le N e_{\delta,N}\)
 and \(c_{\delta,N}\in\mathcal M\). Moreover,
\[
        \tau(c_{\delta,N})
        =
        \tau(c\mathbf{1}_{[\delta,N]}(c))
        \le \tau(c)<\infty,
\]
and therefore \(c_{\delta,N}\in L^1(\mathcal M,\tau)_+\cap\mathcal M\). Since \(0\le b_{\delta,N}\le c_{\delta,N}\), we also have
\(b_{\delta,N}\in L^1(\mathcal M,\tau)_+\cap\mathcal M\).

The estimates
\[
        \delta e_{\delta,N}\le c_{\delta,N}\le N e_{\delta,N}
\]
hold by the functional calculus of \(c\). Therefore \(c_{\delta,N}\) is invertible in the corner \(e_{\delta,N}\mathcal M e_{\delta,N}\).

We next prove the \(L^1\)-convergence. For \(c\), since \(e_{\delta,N}=\mathbf{1}_{[\delta,N]}(c)\), we have
\[
        c-c_{\delta,N}
        =
        c\mathbf{1}_{(0,\delta)\cup(N,\infty)}(c).
\]
Thus
\[
        \|c-c_{\delta,N}\|_1
        =
        \tau\!\left(c\mathbf{1}_{(0,\delta)\cup(N,\infty)}(c)\right).
\]
Since \(c\in L^1(\mathcal M,\tau)_+\), the right-hand side tends to \(0\) as \(\delta\downarrow0\) and \(N\uparrow\infty\).

It remains to show that \(b_{\delta,N}\to b\) in \(L^1\). Write \(e=e_{\delta,N}\). For a positive \(L^1\)-operator \(x\) and a
projection \(e\), we shall use the elementary block estimate
\[
        \|ex(\mathbf{1}-e)\|_1
        \le
        \tau(exe)^{1/2}
        \tau((\mathbf{1}-e)x(\mathbf{1}-e))^{1/2}.
\]
Indeed, this follows by writing
\[
        ex(\mathbf{1}-e)
        =
        (ex^{1/2})(x^{1/2}(\mathbf{1}-e))
\]
and applying the noncommutative Cauchy--Schwarz inequality in \(L^2(\mathcal M,\tau)\).

Applying this estimate to \(x=b\), we obtain
\[
\begin{aligned}
        \|b-ebe\|_1
        &\le
        \|(\mathbf{1}-e)b(\mathbf{1}-e)\|_1
        +\|eb(\mathbf{1}-e)\|_1
        +\|(\mathbf{1}-e)be\|_1                                      \\
        &\le
        \tau((\mathbf{1}-e)b(\mathbf{1}-e))
        +
        2\,\tau(ebe)^{1/2}
           \tau((\mathbf{1}-e)b(\mathbf{1}-e))^{1/2}.
\end{aligned}
\]
Since \(0\le b\le c\), we have
\[
        \tau((\mathbf{1}-e)b(\mathbf{1}-e))
        \le
        \tau((\mathbf{1}-e)c(\mathbf{1}-e)).
\]
Because \(e\) commutes with \(c\),
\[
        \tau((\mathbf{1}-e)c(\mathbf{1}-e))
        =
        \tau(c(\mathbf{1}-e))
        =
        \tau\!\left(c\mathbf{1}_{(0,\delta)\cup(N,\infty)}(c)\right)
        \longrightarrow 0 .
\]
Also \(\tau(ebe)\le\tau(b)\le\tau(c)<\infty\). Therefore
\[
        \|b-e_{\delta,N}be_{\delta,N}\|_1\to0.
\]
This proves the desired \(L^1\)-convergence of \(b_{\delta,N}\) to \(b\).
\end{proof}

We now pass from bounded cases to arbitrary \(L^1(\mathcal{M},\tau)_+\).
The following lemma is the structural difference between
the semifinite setting and matrix algebra, and plays a key role in the proof of the main
theorem.

\begin{lemma}
\label{lem:L1-ando-approx}
Let \(0\le b\le c\) with \(b,c\in L^1(\mathcal M,\tau)_+\).  Then
there exists a net of positive contractions \(G_i\in\mathcal M\)
such that
\[
    0\le G_i\le\mathbf{1},
\]
and, in \(L^1(\mathcal M,\tau)\),
\[
    G_i c G_i\longrightarrow b .
\]

Moreover, if \(a\in L^1(\mathcal M,\tau)_+\) satisfies \(a\le c\),
then the family
\[
    G_i aG_i
\]
is a uniformly integrable family in \(L^1(\mathcal M,\tau)\), hence relatively weakly compact in
\(L^1(\M,\tau)\) and every weak cluster point \(\widetilde a\) satisfies
\[
    0\le\widetilde a\le b,\qquad \tau(\widetilde a)\leq \tau(a).
\]
\end{lemma}
\begin{proof}
We follow the order-preserving spectral truncations from
Proposition~\ref{trunc}.  For
\(0<\delta<N<\infty\), put
\[
    e_{\delta,N}=\mathbf{1}_{[\delta,N]}(c),
\]
and define
\[
    b_{\delta,N}=e_{\delta,N}be_{\delta,N},
    \qquad
    c_{\delta,N}=e_{\delta,N}ce_{\delta,N}.
\]
Then
\[
    0\le b_{\delta,N}\le c_{\delta,N},
\]
and
\[
    b_{\delta,N}\to b,
    \qquad
    c_{\delta,N}\to c
\]
in \(L^1(\M,\tau)\) as \(\delta\downarrow0\) and \(N\uparrow\infty\).
Moreover, inside the finite-trace corner
\(e_{\delta,N}Me_{\delta,N}\), one has
\[
    \delta e_{\delta,N}
    \le
    c_{\delta,N}
    \le
    N e_{\delta,N}.
\]
Thus \(c_{\delta,N}\) is bounded and invertible in
\(e_{\delta,N}\M e_{\delta,N}\).

Applying the bounded geometric-mean construction from
Lemma~\ref{lem:bounded-tracial-DR}, there exists a positive contraction
\[
    G_{\delta,N}\in e_{\delta,N}\M e_{\delta,N}
\]
such that
\[
    G_{\delta,N}c_{\delta,N}G_{\delta,N}=b_{\delta,N}.
\]
Since \(G_{\delta,N}=e_{\delta,N}G_{\delta,N}e_{\delta,N}\), we also have
\[
    G_{\delta,N}cG_{\delta,N}
    =
    G_{\delta,N}c_{\delta,N}G_{\delta,N}
    =
    b_{\delta,N}.
\]
Consequently,
\[
    G_{\delta,N}cG_{\delta,N}\longrightarrow b
\]
in \(L^1(\M,\tau)\).  This proves the first assertion, after directing the
pairs \((\delta,N)\) by \(\delta\downarrow0\) and \(N\uparrow\infty\).

Now assume that \(0\le a\le c\), and set
\[
    a_{\delta,N}=G_{\delta,N}aG_{\delta,N}.
\]
Then
\[
    0\le a_{\delta,N}
    \le
    G_{\delta,N}cG_{\delta,N}
    =
    b_{\delta,N}.
\]

Since \(b_{\delta,N}\to b\) in \(L^1(\M,\tau)\), the family
\((a_{\delta,N})_{\delta,N}\) is uniformly integrable in the
noncommutative \(L^1\)-sense.  Indeed, \(L^1\)-convergent families are
uniformly integrable, and domination by such a family preserves uniform
integrability for positive elements.  Hence, by the noncommutative
Dunford--Pettis criterion for von Neumann algebra preduals, e.g., see \cite[Theorem II.2]{Akemann1967} and \cite{HRS2003}, the set
\[
    \{a_{\delta,N}:\ 0<\delta<N<\infty\}
\]
is relatively weakly compact in \(L^1(\M,\tau)\simeq \M_*\).

Therefore, after passing to a subnet, we may choose
\( a_{i}\rightarrow \widetilde a
\) weakly in $L^1(\M,\tau)$. The weak convergence  \( b_{i}\rightarrow  b
\) holds trivially by its $L^1$-convergence.  Since
$$ 0\leq a_i\leq b_i,$$
we have
$$ a_i\in L^1(\M,\tau)_+,\qquad b_i-a_i\in L^1(\M,\tau)_+.$$
Thus the weak convergence of $a_i$ gives
\[
    b_i-a_i\overset{w}{\longrightarrow} b-\widetilde a.
\]
Since $L^1(\M,\tau)_+$ is a norm closed convex cone, it is weakly closed. This gives
\[
    0\le \widetilde a\le b.
\]

Finally, for each \( a_{\delta,N}\), \(0\le G_i\le \mathbf{1}\) implies
\[
    \tau( a_{\delta,N})
    =
    \tau( G_{\delta,N} a G_{\delta,N})
    \le
    \tau(a).
\]
The weak convergence in
\(L^1\)-sense then gives
\[
    \tau(\widetilde a)\le \tau(a).
\]
This completes the proof.

\end{proof}

The following theorem is the main result in this subsection.
\begin{theorem}
\label{prop:tracial-DR-first-estimate}
Let \(a \in \mathcal{S}(\M,\tau)\) and \(b,r\in L^1(\mathcal M,\tau)_+\) satisfy
\( a\le b+r\). Also set
\[
    \varepsilon=\tau(r)<1 .
\]
Then there exists \(\widetilde a\in L^1(\mathcal M,\tau)_+
\)
such that
\[
    \widetilde a\le b,
    \quad
    \widetilde a\in \mathcal{S}_\leq(\M,\tau),
\]
and
\[
    F(a,\widetilde a)
    \ge
    (1-\varepsilon)^2 .
\]
Consequently,
\[
    P(a,\widetilde a)
    \le
    \sqrt{1-(1-\varepsilon)^4}.
\]
\end{theorem}

\begin{proof}
Set \( c=b+r \). Then \(a\le c\) and \(b\le c\).  By
Lemma~\ref{lem:L1-ando-approx}, choose positive contractions
\(G_i\in\mathcal M\) such that
\[
    G_i cG_i\to b
\]
in \(L^1\), and every \(L^1\)-cluster point of \(G_i aG_i\)
is dominated by \(b\).

For each approximant, the same calculation as in the bounded case
gives
\[
\begin{aligned}
    1-\tau(G_iaG_i)
    &=
    \tau\bigl(a(\mathbf{1}-G_i^2)\bigr)                       \\
    &\le
    \tau\bigl(c(\mathbf{1}-G_i^2)\bigr)                       \\
    &=
    \tau(c)-\tau(G_i cG_i).
\end{aligned}
\]
Passing to the limit along the approximating net gives
\[
    \limsup_i \bigl(\mathbf{1}-\tau(aG_i^2)\bigr)
    \le
    \tau(c)-\tau(b)
    =
    \tau(r)
    =
    \varepsilon .
\]
Thus, after passing to a
subnet if necessary, let
\(\widetilde a\) be a weak cluster point of \(G_i aG_i\) 
\[
    \widetilde a=\lim_k G_{i_k}aG_{i_k}.
\]
By Lemma~\ref{lem:L1-ando-approx}, we get
\[
    0\le\widetilde a\le b,
    \qquad
    \tau(\widetilde a)\le\tau(a)=1.
\]

Notice that
\[
    \tau(aG_{i_k}^2)\ge 1-\varepsilon-o(1).
\]
Hence Proposition~\ref{lem:tracial-gentle} yields
\[
    F(a,G_{i_k} aG_{i_k})
    >
    (1-\varepsilon-o(1))^2.
\]
Now fix $a$ with trace $1$, then $F(a,\cdot)$ is concave and $L^1$-continuous with the second variable. Then for arbitrary $\eta>0$, the set
$$ V_\eta=\{x\in L^1(\M,\tau)_+ : F(a,x)\geq(1-\varepsilon)^2-\eta\}$$
is convex by the concavity of $F$ and norm-closed with the $L^1$-continuity. Hence $V_\eta$ is weakly closed. Passing to the weak
limit point, we obtain
\[
    F(a,\widetilde a)
    \ge
    (1-\varepsilon)^2 .
\]

Finally, since \(a\) is normalized,
\[
    P(a,\widetilde a)
    =
    \sqrt{1-F(a,\widetilde a)^2}
    \le
    \sqrt{1-(1-\varepsilon)^4}.
\]
\end{proof}

Applying Theorem~\ref{prop:tracial-DR-first-estimate} to the
specific excess-mass decomposition, we obtain the following corollary which will be used to estimate
the purified distance.

\begin{corollary}
\label{cor:tracial-smoothing-estimate}
Let \(\rho \in \mathcal{S}(\M,\tau)\) and \(\sigma\in L^1(\mathcal M,\tau)_+\). Also let
\(\gamma>0\).  Set
\[
    q=(\rho-\gamma \sigma)_+,
    \qquad
    \delta=\tau(q).
\]
Then $0\leq \delta \leq 1$.

Moreover, there exists
\(
    \widetilde \rho\in \mathcal{S}_\leq(\mathcal M,\tau)
\)
such that
\(
    \widetilde \rho\le \gamma \sigma,
    \,
    \tau(\widetilde \rho)\le1
\), 
and
\[
    F(\rho,\widetilde \rho)
    \ge
    (1-\delta)^2 .
\]
 In particular,
\[
    P(\rho,\widetilde \rho)
    \le
    \sqrt{1-(1-\delta)^4}
    \le
    2\sqrt{\delta}.
\]
\end{corollary}

\begin{proof}
We use the variational formula
\[
    \tau\bigl((\rho-\gamma\sigma)_+\bigr)
    =
    \sup_{0\le T\le \mathbf{1}}
    \tau\bigl(T(\rho-\gamma\sigma)\bigr),
\]
where the supremum is taken over positive contractions \(T\in\mathcal M\). Since \(\sigma\ge0\), we have
\[
    \tau\bigl(T(\rho-\gamma\sigma)\bigr)
    \le
    \tau(T\rho)
    \le
    \tau(\rho)
    =
    1 .
\]
This proves $0\leq \delta\leq 1$.

Since
\[
    \rho-\gamma \sigma
    =
    (\rho-\gamma \sigma)_+-(\rho-\gamma \sigma)_-,
\]
we have
\[
    \rho
    \le
    \gamma \sigma+(\rho-\gamma \sigma)_+
    =
    \gamma \sigma+q.
\]
If $\delta=1$, set $\widetilde{\rho}=0$. Otherwise apply Theorem~\ref{prop:tracial-DR-first-estimate} with
\[
    a=\rho,\qquad b=\gamma \sigma,\qquad r=q.
\]
This gives \(\widetilde \rho\le\gamma \sigma\),
\(\tau(\widetilde \rho)\le1\), and
\[
    F(\rho,\widetilde \rho)
    \ge
    (1-\delta)^2.
\]
The purified-distance estimate follows from
\[
    P(\rho,\widetilde \rho)
    =
    \sqrt{1-F(\rho,\widetilde \rho)^2}
    \le
    \sqrt{1-(1-\delta)^4}.
\]
Finally, for \(0\le\delta\le1\),
\[
    1-(1-\delta)^4
    =
    \delta(4-6\delta+4\delta^2-\delta^3)
    \le
    4\delta,
\]
which yields
\[
    P(\rho,\widetilde \rho)
    \le 2\sqrt{\delta}.
\]
This completes the proof.
\end{proof}

\subsection{Semifinite Mosonyi--Ogawa formula}
\label{finite-recover}
The purpose of this subsection is to prove the semifinite Mosonyi--Ogawa formula. The following arguments are motivated by the finite-rank approximations in \cite{Mosonyi2022}. We first record the following consequence of Proposition~\ref{prop:haagerup-information-approximation}.

\begin{proposition}
\label{thm:finite-trace-recoverability}
Let \(\rho,\sigma\in L^1(\M,\tau)_+\) with \(\rho,\sigma\neq0\).  Then,
\begin{align*}
    Q_\alpha(\rho\Vert \sigma)
    &=
    \lim_{\operatorname{Proj}_\tau(\M)\ni e\uparrow \1}
    Q_\alpha(e\rho e\Vert e\sigma e),\quad \alpha >1
    \\
    D_\alpha(\rho\Vert \sigma)
    &=
    \lim_{\operatorname{Proj}_\tau(\M)\ni e\uparrow \1}
    D_\alpha(e\rho e\Vert e\sigma e),\quad \alpha\geq 1.
\end{align*}
Here the limits are net limits over the directed set
\(\operatorname{Proj}_\tau(\M)\).
\end{proposition}
\begin{proof}
If $e_i\to \1$, then 
 $$ \left(\bigcup_{i=1}^{\infty} e_i\M e_i\right)^{\prime \prime} =\M.$$ Hence the desired result is just a special case of \ Proposition~\ref{prop:haagerup-information-approximation}.
\end{proof}

The following proposition gives a one-shot estimate.
\begin{proposition}\label{prop:type-II-one-shot-renyi-bound} Let \(h\in \mathcal{S}(\M,\tau) \) and \(k\in L^1(\mathcal M,\tau)_+\). Also let \(\gamma>0\) and \(s\ge0\). Then \begin{align} \label{eq:type-II-projection-upper} \tau\!\left(h\{h>\gamma k\}\right) &\le \gamma^{-s} Q_{1+s}(h\Vert k), \\ \label{eq:type-II-positive-part-upper} \tau\!\left((h-\gamma k)_+\right) &\le \gamma^{-s} Q_{1+s}(h\Vert k). \end{align} Here the case \(s=0\) is understood as the trivial estimate.
\end{proposition}
\begin{proof} The case \(s=0\) holds trivially. Let \(s>0\) and set \[ \alpha=1+s, \qquad E=\{h>\gamma k\}. \] Put \[ p=\tau(hE), \qquad q=\tau(kE). \] Since \(E(h-\gamma k)E\ge0\), we have \[ p-\gamma q = \tau(E(h-\gamma k)E) \ge0. \] Hence \(p\ge\gamma q\), and therefore \[ p\le \gamma^{-s}p^{1+s}q^{-s}, \] with the usual convention when \(q=0\). Now apply the normal two-outcome measurement \[ x\longmapsto \bigl(\tau(Ex),\tau((\1-E)x)\bigr) \] to the pair of normal positive functionals \(\omega_h,\omega_k\). The data-processing inequality for the sandwiched R\'enyi divergence gives \[ p^{1+s}q^{-s} \le p^{1+s}q^{-s} +(1-p)^{1+s}\bigl(\tau(k)-q\bigr)^{-s} \le Q_{1+s}(h\Vert k). \] This proves \eqref{eq:type-II-projection-upper}. Finally, \[ \tau((h-\gamma k)_+) = \tau(E(h-\gamma k)E) = p-\gamma q \le p, \] which gives \eqref{eq:type-II-positive-part-upper}.
\end{proof}

We are now going to prove the semifinite Mosonyi--Ogawa formula. For \(a\in\mathbb R\), define the Legendre transform
\[
    J(a)=
    \inf_{s\ge0}\{\psi(s)-as\}
    =
    \inf_{s\ge0}s\bigl(D_{1+s}(\rho \Vert \sigma)-a\bigr).
\]

\begin{theorem}
\label{thm:finite-tracial-positive-part-hoeffding}
Let \(\rho \in \mathcal{S}(\M,\tau)\) and \( \sigma \in L^1(\M,\tau)_+\). Then for every \(a\in\mathbb R,a\neq D_\infty(\rho\Vert \sigma)\)
and \(t>0\),
\begin{align}\label{pos-part}
    \lim_{n\to\infty}
    \frac1n
    \log
    \tau_n
    \left(
        \rho_n-t 2^{na}\sigma_n
    \right)_+
    =
     \inf_{s\ge0}s\bigl(D_{1+s}(\rho \Vert \sigma)-a\bigr).
\end{align}
Moreover,
\begin{align}\label{pro-part}
    \lim_{n\to\infty}
    \frac1n
    \log
    \tau_n
    \left[
        \rho_n
        \{\rho_n>t 2^{na}\sigma_n\}
    \right]
    =
     \inf_{s\ge0}s\bigl(D_{1+s}(\rho \Vert \sigma)-a\bigr).
\end{align}
\end{theorem}

\begin{proof}
When $a>D_\infty(\rho \Vert \sigma)$, the asserted exponents are $-\infty$. Besides, if $s(\rho)\nleq s(\sigma)$, the asserted exponents are trivially zero. So we only consider $a<D_\infty(\rho \Vert \sigma)$ and $s(\rho)\leq s(\sigma)$.

Then we divide the proof into four steps.

\smallskip
\noindent
\emph{Step 1: the universal upper bound.}
  For $\gamma>0$, applying Proposition~\ref{prop:type-II-one-shot-renyi-bound}, we get
$$ \tau_n(\rho_n-\gamma \sigma_n)_+\leq \gamma^{-s}Q_{1+s}(\rho_n\Vert \sigma_n).$$

By additivity,
\[
    Q_{1+s}(\rho_n\Vert \sigma_n)=Q_{1+s}(\rho\Vert \sigma)^n
    =
    2^{n\psi(s)}.
\]
Taking \(\gamma=t 2^{na}\), we obtain
\[
    \tau_n(\rho_n-t 2^{na}\sigma_n)_+
    \le
    t^{-s}2^{n(\psi(s)-as)}.
\]
Hence
\[
    \limsup_{n\to\infty}
    \frac1n
    \log
    \tau_n(\rho_n-t 2^{na}\sigma_n)_+
    \le
    \inf_{s\ge0}\{\psi(s)-as\}
    =
    J(a).
\]

\smallskip

\noindent
\emph{Step 2: Finite-spectrum result.}
For the ``$\geq$'' side, when  $J(a)=-\infty$ the conclusion is trivial. Hence we can set $J(a)<\infty$ from now on. We first consider a finite algebra $\mathcal{N}$. Assume first that, for some finite $K$,
\[
    k=\sum_{j=1}^K \lambda_j e_j,\quad \sum_{j=1}^K e_j=\mathbf{1},
\]
where \(\lambda_j>0\) are mutually different, the projections \(e_j\) are mutually orthogonal. Then we can define
\[\mathcal{E}_k(\cdot)=\sum_{j=1}^{K} e_j(\cdot)e_j.\]

For $m\in \mathbb{N}$, we define
\[\mathcal E_m=\mathcal{E}_{k^{\otimes m}}\]
as the completely positive  trace-preserving
map onto the commutant of \(k^{\otimes m}\).  It is indeed
the spectral pinching associated with \(k^{\otimes m}\). Then every eigenvalue of \(k^{\otimes m}\) is of the form
\[
    \lambda_1^{n_1}\cdots \lambda_K^{n_K},
    \qquad
    n_j\in\mathbb N_0,\quad \sum_{j=1}^K n_j=m .
\]
Therefore \(k^{\otimes m}\) has at most
\[
    N_m
    \le
    \binom{m+K-1}{K-1}
    \le
    (m+1)^K
\]
different eigenvalues. The standard pinching inequality therefore gives
\[
    x\le N_m\mathcal E_m(x).
\]

Now set
\[
    \widehat\rho_m=\mathcal E_m(\rho^{\otimes m}),
    \qquad
    k_m=k^{\otimes m}.
\]
Then \(\widehat\rho_m\) commutes with \(k_m\).  Define
\[
    \psi_m(s)
    =
    \frac1m
    \log Q_{1+s}(\widehat \rho_m\Vert k_m)= \frac1m
    \log Q_{1+s}(\mathcal{E}_m( \rho_m)\Vert\mathcal{E}_m (k_m)).
\]
By data processing,
\[
    \psi_m(s)\le \psi(s).
\]
On the other hand, since
\[
    \rho^{\otimes m}\le N_m\widehat \rho_m,
\]
the monotonicity and homogeneity of \(Q_{1+s}\) give
\[
    Q_{1+s}(\rho_m\Vert k_m)
    \le
    N_m^{1+s}
    Q_{1+s}(\widehat \rho_m\Vert k_m).
\]
Equivalently,
\[
    \psi(s)
    \le
    \psi_m(s)
    +
    (1+s)\frac{\log N_m}{m}.
\]
We conclude that
\begin{align}\label{double-es1}
   \psi(s)-(1+s)
   \frac{\log N_m}{m} \leq \psi_m(s)\leq \psi(s).
\end{align}
Let
\[
    J_m(a)=\inf_{s\ge0}\{\psi_m(s)-as\}.
\]
Since
\[
    \frac{\log N_m}{m}
    \le
    K\frac{\log(m+1)}{m}
    \longrightarrow0,
\]
combining this with \eqref{double-es1}, we get $$\psi_m(s) \to \psi(s),\quad s\geq 0.$$
For $a<D_\infty(\rho\Vert \sigma)$, Lemma~\ref{l-trans} can be applied, hence one obtains
$$J_m(a)\to J(a).$$

\smallskip

\noindent
\emph{Step 3: Reduction to commutative case.}
Now we fix \(m\).  Since \(\widehat \rho_m\) and \(k_m\) commute, they generate a
commutative finite von Neumann subalgebra of $\mathcal{N}^{\overline{\otimes}m}$, which we denote by $\mathcal{A}_m$. Then there exists a normal \(*\)-isomorphism (see \cite{TakesakiI}) \[ \mathcal A_m\simeq L^\infty(\Omega_m,\mu_m)\] where $\mu_m $ is a probability measure. Thus they may be regarded as the classical densities.  The Cram\'er's theorem, e.g., \cite[Section~2.2]{DemboZeitouni}, implies that with the given pair
\( (\widehat \rho_m,k_m) \), for every \(b\in\mathbb R\) and every \(c>0\),
\[
    \lim_{q\to\infty}
    \frac1{qm}
    \log
    \tau_m^{\otimes q}
    \left(
        \widehat \rho_m^{\otimes q}
        -
        c 2^{qmb}k_m^{\otimes q}
    \right)_+
    =
    J_m(b).
\]

Now write \(n=qm+r\), where \(0\le r<m\).  Consider the positive
trace-preserving map
\[
    \Phi_{m,q,r}
    =
    \mathcal E_m^{\otimes q}
    \otimes
    \tau^{\otimes r}
    :
    L^1(\mathcal N^{\bar\otimes n})
    \to
    L^1((\mathcal N^{\bar\otimes m})^{\bar\otimes q}).
\]
It satisfies
\[
    \Phi_{m,q,r}(\rho_n)=\widehat \rho_m^{\otimes q},
    \qquad
    \Phi_{m,q,r}(k_n)=\tau(k)^r k_m^{\otimes q}.
\]
By the variational formula
\[
    \tau(x_+)=\sup_{0\le T\le\1}\tau(Tx),
\]
and the definition of $\Phi_{m,q,r}$, we get
\[
\begin{aligned}
    \tau_n(\rho_n-t 2^{na}k_n)_+
    &\ge
    \tau_m^{\otimes q}
    \left(
        \rho_m^{\otimes q}-\tau(k)^r t 2^{(qm+r)a}k_m^{\otimes q}
    \right)_+ .
\end{aligned}
\]
The extra factor \(\tau(k)^r\cdot t 2^{ra}\) is constant with respect to \(q\), hence does not
affect the exponential rate.  Letting $q\to \infty$, simultaneously with $n \to \infty$, we get
\[
    \liminf_{n\to\infty}
    \frac1n
    \log
    \tau_n(\rho_n-t 2^{na}k_n)_+
    \ge
    J_m(a).
\]
Then letting \(m\to\infty\), we get
\[
    \liminf_{n\to\infty}
    \frac1n
    \log
    \tau_n(\rho_n-t 2^{na}k_n)_+
    \ge
    J(a).
\]
Together with Step 1, this proves the positive-part formula when
\(k\) has finite-spectrum
decomposition.

\smallskip

\noindent
\emph{Step 4: Semifinite recovery.}
For a semifinite von Neumann algebra $\M$, let $e_{\delta,N}=\1_{[\delta,N]}(\sigma)$ for $ 0<\delta<N<\infty$.

Note that
$$ \tau_n
    \left(
        \rho_n-t 2^{na}\sigma_n^{\otimes n}
    \right)_+\geq \tau_n\left( e_{\delta,N}^{\otimes n}\rho_ne_{\delta,N}^{\otimes n}-t 2^{na}e_{\delta,N}^{\otimes n}\sigma_ne_{\delta,N}^{\otimes n}\right)_+.$$
Hence we can first reduce the question to the finite algebra $\mathcal{N}=e_{\delta,N}\M e_{\delta,N}$.

On such a finite algebra, $ \widetilde{\sigma}=e_{\delta,N}\sigma e_{\delta,N}$ is bounded and invertible. Hence we can find a finite-spectrum $k_\varepsilon\in L^1(\mathcal{N},\tau|_{\mathcal{N}})_+$ such that
$$ 2^{-\varepsilon}k_\varepsilon \leq \widetilde{\sigma}\leq 2^{\varepsilon}k_\varepsilon ,\quad \tau(k_\varepsilon)=\tau(\widetilde{\sigma}). $$
Denoting $\widetilde{\rho}=e_{\delta,N}\rho e_{\delta,N}$ and $\widetilde{\rho}_n=e_{\delta,N}^{\otimes n}\rho^{\otimes n} e_{\delta,N}^{\otimes n}$, we obtain
\[
    \tau_n
    \left(
        \widetilde{\rho}_n-t 2^{n(a+\varepsilon)}k_{\varepsilon}^{\otimes n}
    \right)_+
    \le
    \tau_n
    \left(
        \widetilde{\rho}_n-t 2^{na}\widetilde{\sigma}_n
    \right)_+
    \le
    \tau_n
    \left(
        \widetilde{\rho}_n-t 2^{n(a-\varepsilon)}k_{\varepsilon}^{\otimes n}
    \right)_+ .
\]
We put
$$  J_\varepsilon(a)=
    \inf_{s\ge0}s\bigl(D_{1+s}(\widetilde{\rho} \Vert k_\varepsilon)-a\bigr).$$
Applying the finite-spectrum decomposition conclusion to \(k_\varepsilon\), we obtain
\begin{align}\label{sand-est}
    J_\varepsilon(a+\varepsilon)\leq
    \liminf_{n\to\infty}
    \frac1n
    \log
    \tau_n(\widetilde{\rho}_n-t 2^{na}\widetilde{\sigma}_n)_+
    \le
    J_\varepsilon(a-\varepsilon).
\end{align}
The function $a\mapsto J(a)$ is the infimum of a family of affine functions and is therefore concave. This gives the continuity at $a$ whenever $J(a)<\infty$.

Then, the choice of $k_\varepsilon$ implies that
\[
    \left|
        D_\alpha(\widetilde{\rho}\Vert \widetilde{\sigma})
        -
        D_\alpha(\widetilde{\rho}\Vert k_\varepsilon)
    \right|
    \le
    \varepsilon,
    \qquad \alpha>1.
\]
Consequently, applying Lemma~\ref{l-trans} again, we get
\[
    J_\varepsilon(a)\to J(a)
\]
as \(\varepsilon\downarrow0\).  This gives
$$\liminf_{n\to \infty}\left(
        \rho_n-t 2^{na}\sigma_n^{\otimes n}
    \right)_+\geq \inf_{s\geq 0} s\Big(D_{1+s}(e_{\delta,N}\rho e_{\delta,N}\Vert e_{\delta,N}\sigma e_{\delta,N})-a\Big). $$ 
Finally, let
$$ \varphi_{\delta,N}(s)=sD_{1+s}(e_{\delta,N}\rho e_{\delta,N}\Vert e_{\delta,N}\sigma e_{\delta,N}),\quad \varphi(s)=sD_{1+s}(\rho \Vert \sigma).$$
By Proposition~\ref{thm:finite-trace-recoverability},
$$ \varphi_{\delta,N}(s)\to \varphi (s),\quad s\geq 0.$$
Then we can pass the limit to the semifinite case by Lemma~\ref{l-trans}. This gives the ``$\geq$'' in~\eqref{pos-part}. With \emph{Step 1}, we prove~\eqref{pos-part}.
\smallskip
\noindent
\emph{Conclusion.}
\[
    R_n(t)
    =
    \tau_n(\rho_n-t 2^{na}\sigma_n)_+ \leq \tau_n\left[
        \rho_n\{\rho_n>t 2^{na}\sigma_n\}
    \right]= P_n(t).
\]
Conversely, we have
\[
    t 2^{na}\tau_n(\sigma_n\{\rho_n>t 2^{na}\sigma_n\})
    \le
    \tau_n [\rho_n\{\rho_n>t 2^{na}\sigma_n\}]
    =
    P_n(t).
\]
Therefore
\[
\begin{aligned}
    R_n(t/2)
    &=
    \tau_n
    \left(
        \rho_n-\frac t2 2^{na}\sigma_n
    \right)_+                                      \\
    &\ge
    \tau_n
    \left[
        \left(
            \rho_n-\frac t2 2^{na}\sigma_n
        \right)
        \{\rho_n>t 2^{na}\sigma_n\}
    \right]                                      \\
    &\ge
    \frac12 P_n(t).
\end{aligned}
\]
Hence
\[
    R_n(t)\le P_n(t)\le 2R_n(t/2).
\]
Since the ``positive part'' formula \eqref{pos-part} holds for both \(t\) and \(t/2\),  \eqref{pro-part} follows.
\end{proof}
\begin{remark}\label{philo}
 A related result on $\sigma$-finite von Neumann algebras was recently established by Junge and Laracuente in~\cite{Junge2025}. Their result is actually an analogue of ~\eqref{pos-part}. However, here we also need the ``projection part'' formula~\eqref{pro-part}. We remark that, for~\eqref{pro-part}, the absence of trace brings ambiguity on its formulation for general von Neumann algebras.
\end{remark}

\subsection{Exact exponents of the smoothing quantity}
Combining the semifinite Datta--Renner type estimate
with the semifinite Mosonyi--Ogawa formula, we now get the following theorem:
\begin{theorem}
\label{thm:semifinite-tracial-smoothing-exponent}
Let \((\M,\tau)\) be a semifinite von Neumann algebra. Let
\(\rho \in \mathcal{S}(\M,\tau)\) and \(\sigma\in L^1(\M,\tau)_+\). Then, for every \(r\in \mathbb R\) with $r\neq D_\infty(\rho \Vert \sigma)$,
\[
    \lim_{n\to\infty}
    -\frac1n
    \log
    \Delta(\rho_n\Vert\sigma_n,nr)
    =
    \frac12
    \sup_{s\ge0}
    s\bigl(r-D_{1+s}(\rho\Vert\sigma)\bigr).
\]
\end{theorem}

\begin{proof}
Corollary~\ref{cor:tracial-smoothing-estimate} gives the achievability bound by applying it to
\(\rho^{\otimes n}\), \(\sigma^{\otimes n}\), and \(\gamma_n=2^{nr}\).
The exponent is then obtained from the positive-part formula in Theorem~\ref{thm:finite-tracial-positive-part-hoeffding}.

For the converse, choose an arbitrary  \(\omega_n \in \mathcal{S}_\leq(\M_n,\tau_n)\) such that \(\omega_n\leq2^{nr}\sigma^{\ot n}\). Define
\begin{equation*}
    Q_n=\{\rho^{\ot n}>9\cdot2^{nr}\sigma^{\ot n}\}
\end{equation*}
and
\begin{equation*}
    p_n:=\tau_n(\rho^{\ot n}Q_n),
    \qquad
    q_n:=\tau_n(\omega_n Q_n).
\end{equation*}
Since $Q_n(\rho^{\ot n}-9\cdot2^{nr}\sigma^{\ot n})Q_n\geq0$ and $\omega_n\leq2^{nr}\sigma^{\ot n}$,
\begin{equation}
\label{eq:pq}
    p_n\geq9\cdot2^{nr}\tau_n(\sigma^{\ot n}Q_n)\geq9q_n.
\end{equation}
Apply the binary measurement $\{Q_n,\1_{\M_n}-Q_n\}$.  Data processing inequality then gives
\begin{align*}
    F(\rho^{\ot n},\omega_n)
    &\leq
    F\bigl((p_n,1-p_n),(q_n,\tau_n(\omega_n)-q_n)\bigr) \\
    &\leq \sqrt{p_nq_n}+\sqrt{(1-p_n)(\tau_n(\omega_n)-q_n)} \\
    &\leq \frac{p_n}{3}+\sqrt{1-p_n},
\end{align*}
where \eqref{eq:pq} and $\tau_n(\omega_n)\leq1$ have been used.  Hence
\begin{align*}
    P(\rho^{\ot n},\omega_n)
    &\geq
    \sqrt{1-\left(\frac{p_n}{3}+\sqrt{1-p_n}\right)^2} \\
    &=
    \sqrt{p_n\left(1-\frac{p_n}{9}-\frac{2}{3}\sqrt{1-p_n}\right)} \\
    &\geq
    \sqrt{p_n\left(\frac13-\frac{p_n}{9}\right)}
    \geq \frac{\sqrt2}{3}\sqrt{p_n}.
\end{align*}
Since $\omega_n$ is arbitrary,
\begin{equation}
\label{eq:delta-pn}
    \Deltaq(\rho^{\ot n}\Vert\sigma^{\ot n},nr)
    \geq
    \frac{\sqrt2}{3}\sqrt{p_n}.
\end{equation}
By Theorem~\ref{thm:finite-tracial-positive-part-hoeffding} with $a=r$ and $t=9$,
\begin{equation}
\label{eq:pn-exp}
    \lim_{n\to\infty}\frac1n\log p_n
    =
    \infs s\bigl(D_{1+s}(\rho\Vert\sigma)-r\bigr)
    =-E(r).
\end{equation}
Combining \eqref{eq:delta-pn} and \eqref{eq:pn-exp} yields
\begin{equation*}
\label{eq:conv-final}
    \limsup_{n\to\infty}
    -\frac1n\log
    \Deltaq(\rho^{\ot n}\Vert\sigma^{\ot n},nr)
    \leq
    \frac12E(r).
\end{equation*}
Together with the achievability bound, we complete the proof.
\end{proof}

Since the formula has been established on tracial cases, we extend it to the following result through Haagerup--Junge--Xu reduction.
\begin{theorem}
\label{thm:general-smoothing-exponent}
Let \(\M\) be a von Neumann algebra. Let
\(\rho \in \mathcal{S}(\M)\) and \(\sigma\in \M_*^+\). Then, for every \(r\in \mathbb R\) with $r\neq D_\infty(\rho \Vert \sigma)$,
\[
    \lim_{n\to\infty}
    -\frac1n
    \log
    \Delta(\rho_n\Vert\sigma_n,nr)
    =
    \frac12
    \sup_{s\ge0}
    s\bigl(r-D_{1+s}(\rho\Vert\sigma)\bigr).
\]
Here the term with \(s=0\) is understood as \(0\), and terms with
\(D_{1+s}(\rho\Vert\sigma)=+\infty\) contribute \(-\infty\) inside the
supremum.
\end{theorem}
% \begin{proof}
% Follow the notations in Section~\ref{reduct}. Let
% \[
% \M\subseteq \mathcal{R}=\M \rtimes_\alpha G
% \]
% be the its lifting. Also put
% \[
% \widehat\rho=\rho\circ E,\qquad
% \widehat\sigma=\sigma\circ E.
% \]
% By Proposition~\ref{prop:haagerup-information-approximation}, for every \(\alpha\geq 1\),
% \[
%  D_\alpha^{\mathcal R_m}
%     (\rho_m\Vert\sigma_m)
%     \nearrow
%     D_\alpha^{\mathcal R}
%     (\widehat\rho\Vert\widehat\sigma)
%     =
%     D_\alpha^{\mathcal M}
%     (\rho\Vert\sigma).
% \]
% and, for every fixed \(n\) and \(\lambda\in\mathbb R\), there exists $\varepsilon_{m}^{(n)}\to 0$ such that
% \begin{align*}
% \Delta_{\mathcal{R}_m^{\bar\otimes n}}
% \left(
% \widehat\rho_m^{\otimes n}\Vert
% \widehat\sigma_m^{\otimes n},\lambda
% \right)
% \leq
% \Delta_{\M^{\bar\otimes n}}
% \left(
% \rho^{\otimes n}\Vert
% \sigma^{\otimes n},\lambda
% \right)\leq
% \Delta_{\mathcal{R}_m^{\bar\otimes n}}
% \left(
% \widehat\rho_m^{\otimes n}\Vert
% \widehat\sigma_m^{\otimes n},\lambda
% \right)+\varepsilon_{m}^{(n)}.
% \end{align*}

% Then, applying Lemma~\ref{l-trans} with
% \(
% sD_{1+s}(\rho_m\Vert\sigma_m)
%    \to
% sD_{1+s}(\rho\Vert\sigma)
% \)
% gives the convergence of the corresponding Legendre transforms. For the exponent side, choose \(m=m(n)\) such that
% \(\varepsilon_{m(n),n}\) is exponentially negligible. Thus the approximation implies the convergence. This completes the proof.
% \end{proof}
\begin{proof}
After checking the trivial cases, we may assume that
\(s(\rho)\leq s(\sigma)\) and
\(r<D_\infty(\rho\Vert\sigma)\). Following the notation of Haagerup-Junge-Xu reduction in Section~\ref{reduct}, let \(\rho_m,\sigma_m\in \mathcal{R}_m\) be the approximants of $\rho,\sigma$.

Proposition~\ref{prop:haagerup-information-approximation} gives
\[
\limsup_{n\to\infty}
-\frac1n\log
\Delta(\rho^{\otimes n}\Vert\sigma^{\otimes n},nr)
\leq
\limsup_{n\to\infty}
-\frac1n\log
\Delta(\rho_m^{\otimes n}\Vert\sigma_m^{\otimes n},nr).
\]
and then applying Theorem~\ref{thm:semifinite-tracial-smoothing-exponent} gives
\[
\limsup_{n\to\infty}
-\frac1n\log
\Delta(\rho^{\otimes n}\Vert\sigma^{\otimes n},nr)
\leq
\frac12\sup_{s\geq0}
s\bigl(r-D_{1+s}(\rho_m\Vert\sigma_m)\bigr).
\]
Letting \(m\to\infty\), by martingale
recovery Proposition~\ref{prop:haagerup-information-approximation} and Lemma~\ref{l-trans}, one obtains the required upper bound.

For the lower bound, the one-shot estimate Proposition~\ref{prop:type-II-one-shot-renyi-bound}, together with the data processing inequality, gives
\[
\Delta(\rho_m^{\otimes n}\Vert\sigma_m^{\otimes n},nr)
\leq
2\,2^{-\frac{ns}{2}
(r-D_{1+s}(\rho_m\Vert\sigma_m))}\leq 2\,2^{-\frac{ns}{2}
(r-D_{1+s}(\rho\Vert\sigma))}
\]
uniformly in \(m\) for every \(s\geq0\). By Proposition~\ref{prop:haagerup-information-approximation} again, for each $n$, there exists $\varepsilon_m^{(n)}\to 0$ such that
$$ \Delta(\rho^{\otimes n}\Vert\sigma^{\otimes n},nr)\leq \Delta(\rho_m^{\otimes n}\Vert\sigma_m^{\otimes n},nr)+\varepsilon_m^{(n)}.$$
Choosing a sequence \(m=m(n)\) such that $\varepsilon_{m(n)}^{(n)}\leq 2^{-n^2}$. Then it is negligible on the exponential scale , which yields
\[
\liminf_{n\to\infty}
-\frac1n\log
\Delta(\rho^{\otimes n}\Vert\sigma^{\otimes n},nr)
\geq
\frac12\sup_{s\geq0}
s\bigl(r-D_{1+s}(\rho\Vert\sigma)\bigr).
\]
This completes the proof.
\end{proof}

% \begin{proof}
% After passing to the support corner, we may assume that \(\sigma\) is
% faithful. Let
% \[
%     \widehat\rho=\rho\circ E,\qquad
%     \widehat\sigma=\sigma\circ E,
% \]
% and let \(\rho_m,\sigma_m\) be their restrictions to \(\mathcal R_m\).
% Set
% \[
%     E_m(r)
%     :=
%     \frac12\sup_{s\geq0}
%     s\bigl(r-D_{1+s}(\rho_m\Vert\sigma_m)\bigr).
% \]
% By data processing and Theorem~5,
% \[
%     \limsup_{n\to\infty}
%     -\frac1n
%     \log\Delta(\rho^{\otimes n}\Vert\sigma^{\otimes n},nr)
%     \leq E_m(r),
% \]
% where a harmless perturbation of \(r\) may be used if
% \(r=D_\infty(\rho_m\Vert\sigma_m)\). Proposition~2 and Lemma~1 then give
% \[
%     E_m(r)\longrightarrow
%     \frac12\sup_{s\geq0}
%     s\bigl(r-D_{1+s}(\rho\Vert\sigma)\bigr).
% \]

% Conversely, the one-shot estimate used in the proof of Theorem~5,
% together with
% \[
%     D_{1+s}(\rho_m\Vert\sigma_m)
%     \leq D_{1+s}(\rho\Vert\sigma),
% \]
% gives the reverse bound after choosing, for each \(n\), an index \(m=m(n)\)
% in Proposition~2 so that the approximation error is exponentially
% negligible. This standard diagonal argument completes the proof.
% \end{proof}

\section{Quantum information decoupling in von Neumann algebras}\label{decouple}
We regard \(\M\) as
the reference system and \(\mathcal H\) as the controlled quantum system which is a finite-dimensional Hilbert space. A bipartite state is represented by a normal positive state
\[
    \rho_{\M\mathcal H}
    \in
   \mathcal{S}(\M\bar\otimes \mathcal{B}(\mathcal H)).
\]
Its marginal on the reference algebra is denoted by
\[
    \rho_\M(x)
    =
    \rho_{\M\mathcal H}(x\otimes \1_\mathcal{H}),\quad x\in \M.
\]
For \(\alpha\ge 1\), we define
\[
    I_\alpha(\M:\mathcal H)_\rho
    =
    \inf_{\sigma_{\mathcal H}\in \mathcal{S}(\mathcal H)}
    D_\alpha
    \left(
        \rho_{\M\mathcal H}
        \,\middle\|\,
        \rho_\M\otimes\sigma_{\mathcal H}
    \right).
\]

The purpose of this section is to study catalytic decoupling from a general von Neumann algebraic reference system. In contrast to the finite-dimensional
case, the reference algebra \(\M\) may be an arbitrary von Neumann algebra.

We use two equivalent formulations:

\noindent{\bf Random-unitary formulation.} In the random-unitary formulation,
one adjoins a finite-dimensional catalyst \(\mathcal K\), prepares it in
a state \(\eta_{\mathcal K}\in \mathcal{S}(\mathcal K)\), and applies a random-unitary
channel. In the Heisenberg picture, the random-unitary channel acting on $\mathcal{H}\otimes\mathcal{K}$ is defined as
\[
    \Phi(X)
    =
    \frac1m
    \sum_{i=1}^m
    U_i^* X U_i,
    \qquad
    U_i\in \mathcal{U}(\mathcal H\otimes\mathcal K).
\]
The resulting output state is
$$
  (\operatorname{id}_{\M}\otimes \Phi)_{*} (\rho_{\M \mathcal{H}}\otimes \eta_\mathcal{K}).
$$
The target is a product state
with marginal \(\rho_\M\) on the reference algebra. The cost is \(\log m\).
Thus we define \(P^{\mathrm{dec-u}}_{\M:\mathcal H}(\rho_{\M\mathcal H},r)\)
as the infimum of
\[
    P\left(
        (\operatorname{id}_\M\otimes\Phi)_{*}
        (\rho_{\M\mathcal H}\otimes\eta_{\mathcal K}),
        \rho_\M\otimes\omega_{\mathcal H\mathcal K}
    \right)
\]
over all finite-dimensional catalysts \(\mathcal K\), all catalyst states
\(\eta_{\mathcal K}\in \mathcal{S}(\mathcal K)\), all output states
\(\omega_{\mathcal H\mathcal K}\in \mathcal{S}(\mathcal H\otimes\mathcal K)\), and
all random-unitary channels \(\Phi\) of the above form satisfying
\[
    \log m\le r.
\]

\noindent{\bf Subsystem-removal formulation.} In the subsystem-removal formulation, one adjoins a finite-dimensional
catalyst \(\mathcal K\), applies a unitary
\[
    U:\mathcal H\otimes\mathcal K
    \longrightarrow
    \mathcal H'\otimes\mathcal L,
\]
and discards the finite-dimensional subsystem \(\mathcal L\). The cost
is \(\log\dim\mathcal L\).
% We then define
% \begin{align*}
% \Phi_{U,\mathcal{L}}\;:\; \B(\mathcal{H}^\prime)&\to \B(\mathcal{H}\otimes \mathcal{K}),\\
% x\;\; &\mapsto U^* (x\otimes \1_\mathcal{L}) U.
% \end{align*}
We define
\(P^{\mathrm{dec}}_{M:\mathcal H}(\rho_{M\mathcal H},r)\) as the infimum of
\[ P\big(\Phi_*^{U,\mathcal{L}}( \rho_{\M\mathcal H}\otimes\eta_{\mathcal K}),\rho_\M \otimes \omega_{\mathcal{H}^\prime}\big)\]
over all finite-dimensional catalysts \(\mathcal K\), all catalyst states
\(\eta_{\mathcal K}\in \mathcal{S}(\mathcal K)\), all finite-dimensional output
systems \(\mathcal H'\), all discarded systems \(\mathcal L\) with
\[
    \log\dim\mathcal L\le r,
\]
all states \(\omega_{\mathcal H'}\in \mathcal{S}(\mathcal H')\), and all unitaries
\(U:\mathcal H\otimes\mathcal K\to\mathcal H'\otimes\mathcal L\).
 For $x\in \M $ and $y\in \B(\mathcal{H}^\prime)$, the output state is characterized by
\[
   \Phi_*^{U,\mathcal{L}}(\rho_{\M\mathcal H}\otimes\eta_{\mathcal K})(x\otimes y)=
            (\rho_{\M\mathcal H}\otimes\eta_{\mathcal K})(x\otimes U^* (y\otimes \1_\mathcal{L})U).
\]

The standard equivalence between these two formulations is implemented
entirely on the finite-dimensional side. Hence the same argument as in
the matrix-algebraic setting gives the normalization relation
\begin{align}\label{normalize}
    P^{\mathrm{dec}}_{\M:\mathcal H}
    (
        \rho_{\M\mathcal H},r
    )
    =
    P^{\mathrm{dec-u}}_{\M:\mathcal H}
    (
        \rho_{\M\mathcal H},2r
    ).
\end{align}
See \cite[Proposition~6]{LiYao2024} for the proof.

For tensor products, we consider
\[
    \rho_{\M\mathcal H}^{\otimes n}
    \in
    \left(\left(
        \M
        \bar\otimes
        \mathcal{B}(\mathcal H)\right)^{\bar\otimes n}
    \right)_*^+.
\]
The reliability function is defined by
\[
    E^{\mathrm{dec-u}}_{\M:\mathcal H}(\rho_{\M\mathcal H},2r)
    =
    E^{\mathrm{dec}}_{\M:\mathcal H}(\rho_{\M\mathcal H},r)
    =
    \limsup_{n\to\infty}
    \frac{-1}n
    \log
    P^{\mathrm{dec}}_{\M^{\bar\otimes n}:\mathcal H^{\otimes n}}
    \left(
        \rho_{\M\mathcal H}^{\otimes n},
        nr
    \right).
\]

The main task of this section is to identify this exponent. Our conclusion is stated as follows:
\begin{theorem}
\label{thm:semifinite-decoupling-reliability}
Let \(\M\) be a von Neumann
algebra and \(\mathcal H\) be a finite-dimensional Hilbert space.  Let
\[
    \rho_{\M \mathcal H}
    \in
    \mathcal{S}(
        \M \bar\otimes \mathcal{B}(\mathcal H)
    )
\]
be a normal state. Then, for every \(r\ge0\),
\[
    E^{\mathrm{dec}}_{\M:\mathcal H}
    (
        \rho_{\M\mathcal H},r
    )
    \ge
    \sup_{0\le s\le1}
    s\left(
        r-\frac12 I_{1+s}(\M:\mathcal H)_\rho
    \right),
\]
and
\[
    E^{\mathrm{dec}}_{\M:\mathcal H}
    (
        \rho_{\M\mathcal H},r
    )
    \le
    \sup_{s\ge0}
    s\left(
        r-\frac12 I_{1+s}(\M:\mathcal H)_\rho
    \right),\quad r\neq \frac{1}{2}I_{\infty}(\M :\mathcal{H})_\rho.
\]

Consequently, if
\[
    r\le
    R^\sharp
    =
    \frac12
    \left.
    \frac{d}{ds}
    \left(
        s I_{1+s}(\M:\mathcal H)_\rho
    \right)
    \right|_{s=1},
\]
then the two bounds coincide and
\[
    E^{\mathrm{dec}}_{\M:\mathcal H}
    (
        \rho_{\M\mathcal H},r
    )
    =
    \max_{0\le s\le1}
    s\left(
        r-\frac12 I_{1+s}(\M:\mathcal H)_\rho
    \right).
\]
\end{theorem}
This result has a conceptual consequence. Although the reference system is now allowed to be an arbitrary von Neumann algebra, the reliability law remains valid and can be characterized by the R\'enyi mutual information under the
critical rate.

As a result, the decoupling exponent should be regarded as an intrinsic operator-algebraic quantity, rather than a consequence of finite-dimensional matrix techniques. In this sense, the operational content of catalytic decoupling remains stable when the reference system is replaced by a more complicated von Neumann algebraic system.

\subsection{Generalized layer-cake and convex split lemma}
In this subsection we prove a generalized version of the
Cheng--Gao--Hirche--Huang--Liu inequality, i.e. \cite[Theorem~1]{ChengGaoHircheHuangLiu2025}. A key step in the proof of this inequality is the layer-cake lemma in Cheng--Liu \cite{cheng2025}. However, Cheng--Liu's proof relies on the countability of the spectrum, while such a property does not necessarily hold for general von Neumann algebras. To address this issue, we give Lemma~\ref{zero-level-average}, which was established for semifinite von Neumann algebras in a recent work \cite[Lemma~6.3]{Cheng2026}.

For a bounded invertible positive operator \(X\in \mathcal \M_+\) and a
self-adjoint \(Y\in \mathcal \M\), we denote by
\[
    \D \ln(X)[Y]
    =
    \left.\frac{d}{dt}\right|_{t=0}\ln(X+tY)
    =
    \lim_{t\to0}
    \frac{\ln(X+tY)-\ln(X)}{t}
\]
the Fréchet derivative of the operator logarithm at \(X\) in the
direction \(Y\).  The limit is taken in operator norm.  In this sense,
one has the well-known representation
\[
    \D \ln(X)[Y]
    =
    \int_0^\infty
    (X+\lambda \mathbf{1})^{-1}Y(X+\lambda \mathbf{1})^{-1}\,d\lambda.
\]

We now give several auxiliary results.
\begin{lemma}\label{zero-level-average}
Let $\M$ be a von Neumann algebra, let $X\in \M_+$ be bounded and
invertible, and let $Y=Y^*\in \M$. Put
\[
Z=X^{-1/2}YX^{-1/2},\qquad T_u=Y-uX,
\]
and denote
\[
P_u=\1_{\{0\}}(T_u),\qquad E_u=\1_{\{u\}}(Z).
\]
Then, for every interval \(I\subseteq \mathbb R\),
\[
\int_I\, P_u\,du=0
\]
in the $\sigma$-weak topology.

\end{lemma}

\begin{proof}
We first note that the scalar functions
\[
u\mapsto \varphi(P_u),\qquad \varphi\in \M_*,
\]
are Borel measurable. Indeed, let
\[
f_n(t)=\max\{1-n|t|,0\},\qquad t\in\mathbb R.
\]
Then $f_n\in C_b(\mathbb R)$ and
\[
f_n(T_u)\longrightarrow \1_{\{0\}}(T_u)=P_u
\]
strongly for every $u$. Since $u\mapsto T_u$ is norm-continuous, the map
$u\mapsto f_n(T_u)$ is norm-continuous. Hence
$u\mapsto \varphi(f_n(T_u))$ is continuous, and
$u\mapsto \varphi(P_u)$ is Borel measurable as a pointwise limit.

We next compare $P_u$ with the atomic spectral projection of $Z$ at $u$. Working in a faithful
normal representation, if $T_u\xi=0$, then
\[
0=(Y-uX)\xi=X^{1/2}(Z-u\1)X^{1/2}\xi.
\]
Since $X^{1/2}$ is invertible, this implies
\[
X^{1/2}\xi\in E_u\mathcal H.
\]
Thus
\[
\operatorname{Ran}(P_u)\subseteq X^{-1/2}E_u\mathcal H.
\]
Let $\ell(X^{-1/2}E_u)$ be the left support projection of $X^{-1/2}E_u$.
Then
\[
P_u\leq \ell(X^{-1/2}E_u).
\]
Moreover,
\[
(X^{-1/2}E_u)^*(X^{-1/2}E_u)
=
E_uX^{-1}E_u
\geq \|X\|_\infty^{-1}E_u.
\]
Consequently,
\[
X^{-1/2}E_uX^{-1/2}
\geq
\|X\|_\infty^{-1}\ell(X^{-1/2}E_u),
\]
and hence
\[
P_u\leq
\|X\|_\infty X^{-1/2}E_uX^{-1/2}.
\]

Now fix $\varphi\in \M_*^+$. Define another positive normal functional
\[
\psi(a)=\varphi(X^{-1/2}aX^{-1/2}),\qquad a\in \M.
\]
Then
\[
0\leq \varphi(P_u)
\leq
\|X\|_\infty \psi(E_u)<\infty.
\]
Let $\mu$ be the finite positive Borel measure on $\mathbb R$ defined by
\[
\mu(S)=\psi(\1_S(Z)),\qquad S\subset\mathbb R\ \text{Borel}.
\]
Then
\(
\psi(E_u)=\mu(\{u\}).
\)
The set of atoms of a finite positive Borel measure is at most countable.
Hence \(
u\mapsto \mu(\{u\})
\) is $0$ almost everywhere in the sense of Lebesgue measure. Therefore we get
\[
\int_I\,\mu(\{u\})\,du=0.
\]
It follows that
\[
0\leq
\int_I\,\varphi(P_u)\,du
\leq
\|X\|_\infty
\int_I\,\mu(\{u\})\,du
=0.
\]
Thus
\[
\int_I\,\varphi(P_u)\,du=0,
\qquad \varphi\in \M_*^+.
\]

Let
\[
Q=\int_I P_u\,du
\]
be the $\sigma$-weak integral. Since $P_u\geq0$, we have $Q\geq0$.
The preceding computation gives
\[
\varphi(Q)=0,\qquad \varphi\in \mathcal M_*^+.
\]
Since positive normal functionals separate the positive cone of $\M$, we
obtain $Q=0$. This proves the claim.
\end{proof}
\begin{remark}
In Cheng-Liu~\cite{cheng2025}, they prove the layer-cake lemma for the reason that the ``singular" spectrum is always countable, hence a zero-measure set. However, the previous lemma reveals that, even though the spectrum of a general operator may not be countable, the  ``singular" part still contributes nothing. This property indeed relies on the von Neumann-algebraic structure.
\end{remark}

The following is the general operator layer-cake lemma. The proof is a modified version of the initial proof in \cite{cheng2025}, with the obstacle being overcame by Lemma~\ref{zero-level-average}. 
\begin{lemma}
\label{lemma:layer-cake}
Let \(\M\) be an arbitrary von Neumann algebra.  Let
\(X\in\M_+\) be bounded and  invertible, and let
\(Y=Y^*\in\M\).  Then
\[
    \D \ln(X)[Y]
    =
    \int_0^\infty \{uX<Y\}\,du
    -
    \int_{-\infty}^0 \{uX>Y\}\,du .
\]
In particular, if \(Y\ge0\), then
\[
    \D \ln(X)[Y]
    =
    \int_0^\infty \{uX<Y\}\,du .
\]
\end{lemma}

\begin{proof}
Put
\[
    Z=X^{-1/2}YX^{-1/2}.
\]
Choose \(r>0\) so large that
\[
    r\1>2Z>-r\1 .
\]
Then \(Y-uX<0\) for \(u\ge r\), and \(Y-uX>0\) for \(u\le -r\).  Hence the
claimed identity is equivalent to
\[
    \D \ln(X)[Y]
    =
    \int_0^r \{Y-uX>0\}\,du
    -
    \int_{-r}^0 \{Y-uX<0\}\,du .
\]
We write \(T_u=Y-uX\).

For \(\eps>0\), define the scalar functions
\[
    \Pi_\eps^+(x)
    =
    \frac12\left(1+\frac2\pi\arctan\frac{x}{\eps}\right),
    \qquad
    \Pi_\eps^-(x)
    =
    \frac12\left(1-\frac2\pi\arctan\frac{x}{\eps}\right).
\]
Equivalently, adopting the principal branch of \(\ln\), which we denote by $\Ln$, we get
\[
    \Pi_\eps^+(x)
    =
    \frac{1}{2\pi i}\bigl(\Ln(ix+\eps)-\Ln(-ix+\eps)+\pi i\bigr),
\]
and
\[
    \Pi_\eps^-(x)
    =
    \frac{1}{2\pi i}\bigl(-\Ln(ix+\eps)+\Ln(-ix+\eps)+\pi i\bigr).
\]
By the continuous functional calculus, for each fixed \(\eps>0\),
\[
\begin{aligned}
&\int_0^r \Pi_\eps^+(T_u)\,du
  -\int_{-r}^0 \Pi_\eps^-(T_u)\,du                                     \\
&\quad =
\frac{1}{2\pi i}
\left[
    \int_{-r}^r \Ln(i(Y-uX)+\eps\1)\,du
    -
    \int_{-r}^r \Ln(-i(Y-uX)+\eps\1)\,du
\right].
\end{aligned}
\]
Here the integrals are norm Bochner integrals, because the integrands are
norm-continuous in \(u\).

Let
\[
    C_r^+=\{r\mathrm{e}^{i\theta}:0\leq\theta\leq\pi\},
    \qquad
    C_r^-=\{r\mathrm{e}^{i\theta}:-\pi\leq\theta\leq0\},
\]
with the indicated counterclockwise orientations. If
\(\operatorname{Im}z>-\eps/(2\norm{X}_\infty)\), then, since
\(0<X\leq \norm{X}_\infty\1\),
\[
\begin{aligned}
    \operatorname{Re}(i(Y-zX)+\eps\1)
    =
    (\operatorname{Im}z)X+\eps\1                            \geq
    \bigl(\eps+\operatorname{Im}z\,\norm{X}_\infty\bigr)\1
    >\frac{\eps}{2}\1 .
\end{aligned}
\]
Thus the first logarithm is holomorphic on the open strip
\(\operatorname{Im}z>-\eps/(2\norm{X}_\infty)\), which contains the
closed upper half-disk bounded by \([-r,r]\) and \(C_r^+\).  Similarly, if
\(\operatorname{Im}z<\eps/(2\norm{X}_\infty)\), then
\[
    \operatorname{Re}(-i(Y-zX)+\eps\1)>\frac{\eps}{2}\1,
\]
and this strip contains the closed lower half-disk bounded by \([-r,r]\) and
\(C_r^-\).  Consequently the maps
\[
    z\mapsto \Ln(i(Y-zX)+\eps\1),
    \qquad
    z\mapsto \Ln(-i(Y-zX)+\eps\1)
\]
are holomorphic on the corresponding open half-neighbourhoods.  Cauchy's
integral theorem gives
\[
    \int_{-r}^r \Ln(i(Y-uX)+\eps\1)\,du
    =
    -\int_{C_r^+}\Ln(i(Y-zX)+\eps\1)\,dz,
\]
and
\[
    -\int_{-r}^r \Ln(-i(Y-uX)+\eps\1)\,du
    =
    -\int_{C_r^-}\Ln(-i(Y-zX)+\eps\1)\,dz.
\]
Consequently,
\begin{equation}
\label{eq:regularized-arc}
\begin{aligned}
&\int_0^r \Pi_\eps^+(T_u)\,du
  -\int_{-r}^0 \Pi_\eps^-(T_u)\,du                                     \\
&\quad =
-
\frac{1}{2\pi i}
\left[
    \int_{C_r^+}\Ln(i(Y-zX)+\eps\1)\,dz
    +
    \int_{C_r^-}\Ln(-i(Y-zX)+\eps\1)\,dz
\right].
\end{aligned}
\end{equation}

We now let \(\eps\downarrow0\).  We first consider the left-hand side.  For
all real \(x\),
\[
    \Pi_\eps^+(x)\longrightarrow \1_{(0,\infty)}(x)+\frac12\1_{\{0\}}(x),
    \quad
    \Pi_\eps^-(x)\longrightarrow \1_{(-\infty,0)}(x)+\frac12\1_{\{0\}}(x)
\]
pointwise. Let \(P_u=\1_{\{0\}}(T_u)\).  After testing against any positive normal
functional \(\varphi\in\M_*^+\), bounded convergence gives
\[
\begin{aligned}
\int_0^r \varphi(\Pi_\eps^+(T_u))\,du
&\longrightarrow
\int_0^r \varphi(\{T_u>0\})\,du
+\frac12\int_0^r \varphi(P_u)\,du,                                     \\
\int_{-r}^0 \varphi(\Pi_\eps^-(T_u))\,du
&\longrightarrow
\int_{-r}^0 \varphi(\{T_u<0\})\,du
+\frac12\int_{-r}^0 \varphi(P_u)\,du .
\end{aligned}
\]
Lemma~\ref{zero-level-average} shows that $\int_I \varphi(P_u)du$ is zero for $I=[-r,0]$ and $I=[0,r]$. Hence
\begin{equation*}
\int_0^r \Pi_\eps^+(T_u)\,du
  -\int_{-r}^0 \Pi_\eps^-(T_u)\,du
\xrightarrow[\eps\downarrow0]{\sigma-w}
\int_0^r\{T_u>0\}\,du-
\int_{-r}^0\{T_u<0\}\,du .
\end{equation*}
This is precisely the point where the finite-dimensional null-set argument in
Cheng--Liu is replaced.

Besides, we have
\begin{align*}
    -\frac{1}{2\pi i}
    \Big[\int_{C_r^+}\Ln(i(Y-zX)+\eps\1)\,dz
    +
    \int_{C_r^-}&\Ln(-i(Y-zX)+\eps\1)\,dz\Big]
 \\
&\to \D \ln(X)[Y].
\end{align*}
This operator-norm convergence follows from the calculation in \cite[version 1, Lemma B.2, from (132)]{cheng2025}, also see \cite[Theorem 6.1, (135)-(159)]{Cheng2026}.
% For the reader's convienience, we include  the caluculation in Appendix~\ref{app:contour-calculation}.
Combining these limits in \eqref{eq:regularized-arc}, we prove the result.
\end{proof}

\begin{proposition}
\label{lem:finite-corner-change-variable}
Let \((\mathcal N,\tau)\) be a finite von Neumann algebra.  Let \(A,B\in\mathcal N_+\), and assume that
\(B\ge \delta \mathbf{1}\) for some \(\delta>0\).  Then, for every \(0<s\le1\),
\[
    \int_0^\infty
    \tau\!\left[A\{A>\gamma B\}\right]\gamma^{s-1}\,d\gamma
    =
    \int_0^\infty
    \tau\!\left[
        \left(
            (B+t\mathbf{1})^{-1/2}A(B+t\mathbf{1})^{-1/2}
        \right)^{1+s}
    \right]\,dt .
\]
\end{proposition}
\begin{proof}
Let $r=\Vert B^{-1/2}AB^{-1/2}\Vert$. Since \(A\le rB\), we have
\[
    \{A>\gamma B\}=0,\qquad \forall \gamma>r.
\]
Hence the left-hand side is an integral over \([0,r]\).

For \(t\ge0\), denote
\[
    S_t=B+t\mathbf{1},
    \qquad
    K_t=S_t^{-1/2}AS_t^{-1/2}.
\]
Then \(0\le K_t\le r\mathbf{1}\).  We first prove
\[
    \int_0^r \{A>\gamma B\}h(\gamma)\,d\gamma
    =
    \int_0^\infty
    S_t^{-1/2}K_t h(K_t)S_t^{-1/2}\,dt
\]
for every continuous function \(h\in C[0,r]\), where the integrals are
understood ultraweakly.

By Lemma~\ref{lemma:layer-cake}, for every \(\beta\ge0\) we have
\[
\begin{aligned}
\D\ln(B+\beta A)[A]
    &=
    \int_0^\infty \{u(B+\beta A)<A\}\,du                                      \\
    &=
    \int_0^{1/\beta}
        \{(1-u\beta)A>uB\}\,du                                                  \\
    &=
    \int_0^\infty
        \{A>\gamma B\}\frac{1}{(1+\beta\gamma)^2}\,d\gamma                       \\
    &=
    \int_0^r
        \{A>\gamma B\}\frac{1}{(1+\beta\gamma)^2}\,d\gamma .
\end{aligned}
\]
Here we used the change of scalar variables
\[
    \gamma=\frac{u}{1-u\beta},
    \qquad
    du=\frac{1}{(1+\beta\gamma)^2}\,d\gamma .
\]
On the other hand, for general nonnegative operators, \cite[Fact (vi)]{cheng2025} gives
\[
\begin{aligned}
   \D\ln(B+\beta A)[A]&=\int_0^\infty
    (B+\beta A+t\mathbf{1})^{-1}A(B+\beta A+t\mathbf{1})^{-1} dt
    \\
    &=
    \int_0^\infty
    S_t^{-1/2}(\mathbf{1}+\beta K_t)^{-1}K_t(\mathbf{1}+\beta K_t)^{-1}S_t^{-1/2}dt                            \\
    &=
    \int_0^\infty
    S_t^{-1/2}
    K_t(\mathbf{1}+\beta K_t)^{-2}
    S_t^{-1/2}\,dt .
\end{aligned}
\]
Here in the second line we use the referred fact and the third line follows by the commutativity between $(1+\beta K_t)$ and  $K_t$. Therefore, for every \(\beta\ge0\),
\[
    \int_0^r
        \{A>\gamma B\}\frac{1}{(\mathbf{1}+\beta\gamma)^2}\,d\gamma
    =
    \int_0^\infty
        S_t^{-1/2}
        K_t(\mathbf{1}+\beta K_t)^{-2}
        S_t^{-1/2}\,dt .
\]
Testing against arbitrary normal functionals on \(\mathcal N\), we may
differentiate both sides with respect to \(\beta\).  Since
\(0\le\gamma\le r\) and \(0\le K_t\le r\mathbf{1}\), the differentiations are
justified by dominated convergence.  Taking the \(m\)-th derivative and
then putting \(\beta=0\), we obtain, for every \(m\in\mathbb Z_{+}\),
\[
    \int_0^r
        \{A>\gamma B\}\gamma^m\,d\gamma
    =
    \int_0^\infty
        S_t^{-1/2}K_t^{m+1}S_t^{-1/2}\,dt .
\]
By linearity, this proves the desired formula for
polynomials \(h\).  By the Weierstrass approximation theorem, every
\(h\in C[0,r]\) is uniformly approximated by polynomials.  Since
\[
    \left\|
    \int_0^r \{A>\gamma B\}(h(\gamma)-p(\gamma))\,d\gamma
    \right\|_\infty
    \le r\|h-p\|_\infty ,
\]
and
\[
\begin{aligned}
    &\left\|
    \int_0^\infty
        S_t^{-1/2}K_t(h(K_t)-p(K_t))S_t^{-1/2}\,dt
    \right\|_\infty                                      \\
    &\qquad\le
    \|h-p\|_\infty
    \left\|
    \int_0^\infty S_t^{-1/2}K_tS_t^{-1/2}\,dt
    \right\|_\infty                                      \\
    &\qquad=
    \|h-p\|_\infty\,
    \|\mathcal{D}\ln(B)[A]\|_\infty
    \le
    r\|h-p\|_\infty ,
\end{aligned}
\]
the formula extends to all continuous \(h\in C[0,r]\).

We now apply this formula to the function
\[
    h(\gamma)=\gamma^{s-1}.
\]
When \(s=1\), this is continuous and the conclusion follows directly.
Assume \(0<s<1\).  For \(n\ge1\), define the continuous bounded functions
\[
    h_n(\gamma)=\min\{\gamma^{s-1},n\},
    \qquad 0\le\gamma\le r,
\]
with the convention \(h_n(0)=n\).  Then \(h_n\uparrow \gamma^{s-1}\) on
\((0,r]\).  Hence, by normal monotone convergence,
\[
    \int_0^r
        \{A>\gamma B\}\gamma^{s-1}\,d\gamma
    =
    \lim_{n\to\infty}
    \int_0^r
        \{A>\gamma B\}h_n(\gamma)\,d\gamma ,
\]
and
\[
    \int_0^\infty
        S_t^{-1/2}K_t^sS_t^{-1/2}\,dt
    =
    \lim_{n\to\infty}
    \int_0^\infty
        S_t^{-1/2}K_t h_n(K_t)S_t^{-1/2}\,dt .
\]
Thus
\[
    \int_0^r
        \{A>\gamma B\}\gamma^{s-1}\,d\gamma
    =
    \int_0^\infty
        S_t^{-1/2}K_t^sS_t^{-1/2}\,dt .
\]
Multiplying by \(A\) and taking trace on both sides, we obtain
\[
\begin{aligned}
    \int_0^\infty
    \tau\!\left[A\{A>\gamma B\}\right]\gamma^{s-1}\,d\gamma                 =
    \int_0^\infty
    \tau\!\left[
        A S_t^{-1/2}K_t^sS_t^{-1/2}
    \right]\,dt .
\end{aligned}
\]
Since
\[
    A=S_t^{1/2}K_tS_t^{1/2},
\]
the traciality of \(\tau\) gives
\[
    \tau\!\left[
        A S_t^{-1/2}K_t^sS_t^{-1/2}
    \right]
    =
    \tau\!\left(
        K_t^{1+s}
    \right).
\]
Therefore
\[
    \int_0^\infty
    \tau\!\left[A\{A>\gamma B\}\right]\gamma^{s-1}\,d\gamma
    =
    \int_0^\infty
    \tau\!\left(K_t^{1+s}\right)\,dt ,
\]
which is exactly the result after inserting $K_t=S_t^{-1/2}AS_t^{-1/2}$.
\end{proof}

Here and in the sequel, we define 
$$ c_s=s^s(1-s)^{1-s},\quad 0<s\leq 1, $$
with $c_1=1$.

\begin{proposition}
\label{prop:finite-corner-CG}
  Let $(\mathcal{N},\tau)$ be a finite von Neumann algebra and
\[
    A\in L^1(\mathcal{N},\tau)_+,\qquad
    B\in L^1(\mathcal{N},\tau)_+\cap \mathcal{N},
\]
and assume that \(B\ge \delta \mathbf{1}\) for some \(\delta>0\).  Then, for
every \(0<s\le1\),
\[
\begin{aligned}
    \tau\!\left[
        A\bigl(\ln(A+B)-\ln B\bigr)
    \right]
    &\le
    c_s
    \int_0^\infty
    \tau\!\left[
        \left(
            (B+t \mathbf{1})^{-1/2}A(B+t \mathbf{1})^{-1/2}
        \right)^{1+s}
    \right]\,dt                                      \\
    &\le
    \frac{c_s}{s}\,
    \tau\!\left[
        \left(
            B^{-\,\frac{s}{2(1+s)}}
            A
            B^{-\,\frac{s}{2(1+s)}}
        \right)^{1+s}
    \right].
\end{aligned}
\]
\end{proposition}

\begin{proof}
We first assume \(A\in \mathcal{N}\).  By the fundamental theorem of calculus,
\[
    \ln(A+B)-\ln B
    =
    \int_0^1 \D\ln(B+\beta A)[A]\,d\beta .
\]
Since \(B+\beta A\) is bounded and invertible, Lemma
\ref{lemma:layer-cake} gives
\[
    \D\ln(B+\beta A)[A]
    =
    \int_0^\infty \{u(B+\beta A)<A\}\,du .
\]
Applying the change of variables
\(
    \gamma=\frac{u}{1-\beta u}
\)
yields
\[
    \ln(A+B)-\ln B
    =
    \int_0^\infty \frac{1}{1+\gamma}\{A>\gamma B\}\,d\gamma .
\]
Young's inequality gives
\[
    \frac{1}{1+\gamma}
    \le
   s^s(1-s)^{1-s}\gamma^{s-1},
    \qquad \gamma>0 .
\]
Multiplying by \(A\), taking the trace, and using Proposition
\ref{lem:finite-corner-change-variable}, we obtain
\[
\begin{aligned}
    \tau\!\left(
        A\bigl(\ln(A+B)-\ln B\bigr)
    \right)
    &\le
    c_s
    \int_0^\infty
    \tau\!\left[A\{A>\gamma B\}\right]\gamma^{s-1}\,d\gamma        \\
    &=
    c_s
    \int_0^\infty
    \tau\!\left[
        \left(
            (B+t \mathbf{1})^{-1/2}A(B+t \mathbf{1})^{-1/2}
        \right)^{1+s}
    \right]\,dt .
\end{aligned}
\]
Using the Araki--Lieb--Thirring inequality (see \cite{Araki,Kosaki}), we get
\[
    \int_0^\infty
    \tau\!\left[
        \left(
            (B+t \mathbf{1})^{-1/2}A(B+t \mathbf{1})^{-1/2}
        \right)^{1+s}
    \right]\,dt
    \le
    \frac{1}{s}\,
    \tau\!\left[
        \left(
            B^{-\,\frac{s}{2(1+s)}}
            A
            B^{-\,\frac{s}{2(1+s)}}
        \right)^{1+s}
    \right].
\]
This proves the claim for bounded \(A\).  For general
\(A\in L^1(\mathcal{N},\tau)_+\), apply the bounded case to
\[
    A_{\delta,N}=A\,\mathbf{1}_{[\delta,N]}(A).
\]
The passage to the limit follows from Proposition~\ref{trunc} with $\delta\to 0$, $N\to\infty$ and $L^1$-recoverability for all the related quantities in the inequality.  Hence the same
inequality holds for arbitrary \(A\in L^1(\mathcal{N},\tau)_+\).
\end{proof}

We can now give the semifinite version of the Cheng--Gao--Hirche--Huang--Liu inequality.
\begin{theorem}
\label{thm:semifinite-CG}
Let \(A,B\in L^1(\M,\tau)_+\), and assume that
\(
    s(A)\le s(B).
\)
Then, for every \(0<s\le1\),
\[
\begin{aligned}
    \tau\!\left[
        A\bigl(\ln(A+B)-\ln B\bigr)
    \right]
    % &\le
    % c_s
    % \int_0^\infty
    % \tau\!\left[
    %     \left(
    %         (B+t\,p)^{-1/2}
    %         A
    %         (B+t\,p)^{-1/2}
    %     \right)^{1+s}
    % \right]\,dt                                      \\
    % &\le
    \leq \frac{c_s}{s}\,
    \tau\!\left[
        \left(
            B^{-\,\frac{s}{2(1+s)}}
            A
            B^{-\,\frac{s}{2(1+s)}}
        \right)^{1+s}
    \right].
\end{aligned}
\]
\end{theorem}

\begin{proof}
For \(0<\delta<N<\infty\), set
\[
    e_{\delta,N}=\mathbf{1}_{[\delta,N]}(B),
    \qquad
    A_{\delta,N}=e_{\delta,N}Ae_{\delta,N},
    \qquad
    B_{\delta,N}=e_{\delta,N}Be_{\delta,N}.
\]
Since \(B\in L^1(\M,\tau)_+\), one has \(\tau(e_{\delta,N})<\infty\). Moreover, we have
\[
    \delta e_{\delta,N}
    \le
    B_{\delta,N}
    \le
    N e_{\delta,N}.
\]
Thus \(B_{\delta,N}\) is boundedly invertible in the finite-trace
corner.  Also, we have
\[
    A_{\delta,N}\to A,
    \qquad
    B_{\delta,N}\to B
    \quad\text{in }L^1(\M,\tau)
\]
as \(\delta\downarrow0\) and \(N\uparrow\infty\).

Applying Proposition \ref{prop:finite-corner-CG} in the corner
\(e_{\delta,N}Me_{\delta,N}\), we obtain
\[
\begin{aligned}
    &\tau\!\left[
        A_{\delta,N}
        \bigl(
            \ln (A_{\delta,N}+B_{\delta,N})
            -
            \ln B_{\delta,N}
        \bigr)
    \right]                                              \\
    &\quad\le
    c_s
    \int_0^\infty
    \tau\!\left[
        \left(
            (B_{\delta,N}+t e_{\delta,N})^{-1/2}
            A_{\delta,N}
            (B_{\delta,N}+t e_{\delta,N})^{-1/2}
        \right)^{1+s}
    \right]\,dt                                          \\
    &\quad\le
    \frac{c_s}{s}\,
    \tau\!\left[
        \left(
            B_{\delta,N}^{-\,\frac{s}{2(1+s)}}
            A_{\delta,N}
            B_{\delta,N}^{-\,\frac{s}{2(1+s)}}
        \right)^{1+s}
    \right].
\end{aligned}
\]
It remains to pass to the limit.  The convergence of the first line follows from the recoverability of $D_1$ and the
identity
\[
    \tau\!\left(
        A\bigl(\ln (A+B)-\ln B\bigr)
    \right)
    =
    \frac{1}{\log e}\Big(D(A+B||B)+D(B||A+B)\Big).
\]
The last line is naturally in R\'enyi divergence form. Hence the standard recoverability in Section~\ref{finite-recover} works here. This completes the proof.
\end{proof}

Then we can extend this result to general von Neumann algebras  via Haagerup--Junge--Xu reduction.
\begin{theorem}\label{general-convexsplit}
Let \(\M\) be a von Neumann algebra. Let \(\rho\in \mathcal{S}(\M)\) and  \(0\neq \sigma\in \M_*^+\) satisfy
\(
s(\rho)\leq s(\sigma)
\). 
Then, for every \(\lambda>0\) and \(0< s\leq 1\),
\[
D(\rho\Vert\sigma)
-
D(\rho\Vert\sigma+\lambda\rho)
\leq
\frac{c_s(\log e)}{s}\lambda^s
Q_{1+s}(\rho\Vert\sigma)=
\frac{c_s(\log e)}{s}\lambda^s
2^{\,sD_{1+s}(\rho\Vert\sigma)}.
\]
\end{theorem}

\begin{proof}
If
\(
Q_{1+s}(\rho\Vert\sigma)=+\infty,
\)
there is nothing to prove. Hence we assume that this quantity is
finite.

We still follow the notation in Section~\ref{reduct}. Let
\[
\M\subseteq \M\rtimes_\alpha G=\RR=\left(
        \bigcup_{m\geq 1}\mathcal R_m
    \right)^{\prime \prime}
    =
    \mathcal R.
\]
be the corresponding Haagerup--Junge--Xu reduction algebra. Put
\[
\widehat\rho=\rho\circ E,
\qquad
\widehat\sigma=\sigma\circ E.
\]

Let
\[
\rho_m=\widehat\rho|_{\RR_m},
\qquad
\sigma_m=\widehat\sigma|_{\RR_m}.
\]
Since \(\RR_m\) is finite, Theorem~\ref{thm:semifinite-CG} gives, by putting $A=\lambda\rho_m$ and $B=\sigma_m$,
\[
D(\rho_m\Vert\sigma_m)
-
D(\rho_m\Vert\sigma_m+\lambda\rho_m)
\leq
\frac{(\log e)c_s}{s}\lambda^s
Q_{1+s}(\rho_m\Vert\sigma_m).
\]
Letting \(m\to\infty\) and using the martingale recovery of the
Araki relative entropy and sandwiched R\'enyi divergences yields the same inequality
for \((\widehat\rho,\widehat\sigma)\) on \(\RR\).  The invariance of the information theoretic quantities between $\RR$ and $\M$ then gives the assertion on \(\M\).
\end{proof}

The preceding theorem directly gives the following convex-splitting estimate.
\begin{proposition}
\label{prop:dimension-free-convex-split}
Let \(\M\) be a von Neumann
algebra and \(\mathcal H\) be a finite-dimensional Hilbert space, and let
\[
    \rho_{\M \mathcal H}
    \in
    \mathcal{S}(\M \bar\otimes \mathcal{B}(\mathcal H))
\]
be a normalized state.  Put
\[
    \rho_{\M}(x)
    =
   \rho_{\M \mathcal H}(x\otimes \1_\mathcal{H}),\quad x\in \M,
\]
and fix a state \(\sigma_{\mathcal H}\in \mathcal{S}(\mathcal H)\).  Define
\[
    \theta_{\M\mathcal H}
    =
    \rho_{\M}\otimes\sigma_{\mathcal H}.
\]
For \(m\in\N \), set
\[
    \Theta_m
    =
    \rho_{\M}\otimes\sigma_{\mathcal H}^{\otimes m},
\]
and
\[
    \Omega_m
    =
    \frac1m
    \sum_{j=1}^m
    \rho_{\M\mathcal H_j}
    \otimes
    \bigotimes_{i\ne j}\sigma_{\mathcal H_i}.
\]
Then, for every \(0<s\le1\),
\[
    D(\Omega_m\Vert\Theta_m)
    \le
    \frac{c_s \log(e)}{s}\,
    2^{-s\left(
        \log m
        -
        D_{1+s}
        (
            \rho_{\M\mathcal H}
            \Vert
            \rho_{\M}\otimes\sigma_{\mathcal H}
        )
    \right)}.
\]
\end{proposition}
\begin{proof}
If
\(
    D_{1+s}
    (
        \rho_{\M\mathcal H}
        \Vert
        \theta_{\M\mathcal H}
    )
    =
    +\infty
\),
there is nothing to prove.  Hence we assume that this quantity is finite, which implies
\(
    s(\rho_{\M\mathcal H})
    \le
    s(\theta_{\M\mathcal H})
\).

For \(1\le j\le m\), write
\[
    A_j
    =
    \rho_{\M\mathcal H_j}
    \otimes
    \bigotimes_{i\ne j}\sigma_{\mathcal H_i}.
\]
Thus
\[
    \Omega_m=\frac1m\sum_{j=1}^m A_j.
\]
Let \(\mathcal E_j\) be the quantum channel which keeps
\(\M\mathcal H_j\), discards all other \(\mathcal H\)-registers,
and appends \(\sigma_{\mathcal H}\) on them:
\[
\mathcal E_j(\omega)
=
\omega_{\mathcal M\mathcal H_j}
\otimes
\bigotimes_{i\neq j}\sigma_{\mathcal H_i},
\]
where \(\omega_{\mathcal M\mathcal H_j}\) is the marginal of \(\omega\) on \(\M\mathcal H_j\). Hence we have
\[
\begin{aligned}
\mathcal E_j(\omega)
\left(
x\otimes y_1\otimes\cdots\otimes y_m
\right)
&=
\omega
\left(
x\otimes
\1\otimes\cdots\otimes y_j\otimes\cdots\otimes\1
\right)\times
\prod_{i\ne j}\sigma_{\mathcal H}(y_i).
\end{aligned}
\]
% \[
%     \mathcal E_j(X)
%     =
%     X\left(x\otimes \bigotimes_{i\ne j}\1_{\mathcal H_i} \right)
%     \otimes
%     \bigotimes_{i\ne j}\sigma_{\mathcal H_i},\quad x\in \M\bar{\otimes}\mathcal{H}_j .
% \]
Then
\[
    \mathcal E_j(A_j)=A_j,\qquad
    \mathcal E_j(A_i)=\Theta_m,\quad i\ne j.
\]
Therefore
\[
    \mathcal E_j(\Omega_m)
    =
    \frac1m A_j+\frac{m-1}{m}\Theta_m
    \le
    \frac1m A_j+\Theta_m.
\]

By Donald's identity for the Araki relative entropy,
\[
D(\Omega_m\Vert\Theta_m)
=
\frac1m\sum_{j=1}^m
\left[
D(A_j\Vert\Theta_m)
-
D(A_j\Vert\Omega_m)
\right].
\]
The data-processing inequality under \(\mathcal E_j\) gives
\[
D(A_j\Vert\Omega_m)
\geq
D\bigl(
\mathcal E_j(A_j)
\Vert
\mathcal E_j(\Omega_m)
\bigr)
=
D\bigl(
A_j
\Vert
\mathcal E_j(\Omega_m)
\bigr).
\]
Since the relative entropy is monotone with respect to the order of normal positive functionals, the inequality
\[
    \mathcal E_j(\Omega_m)
    \le
    \frac1m A_j+\Theta_m
\]
implies
\[
D\bigl(
A_j
\Vert
\mathcal E_j(\Omega_m)
\bigr)
\geq
D\left(
A_j
\middle\Vert
\frac1m A_j+\Theta_m
\right).
\]
Consequently,
\[
D(\Omega_m\Vert\Theta_m)
\leq
\frac1m\sum_{j=1}^m
\left[
D(A_j\Vert\Theta_m)
-
D\left(
A_j
\middle\Vert
\Theta_m+\frac1m A_j
\right)
\right].
\]
Since the remaining \(\mathcal{H}\)-registers are the same, we have
\[
D(A_j\Vert\Theta_m)
=
D(\rho_{\M\mathcal{H}}\Vert\theta_{\M\mathcal{H}}),
\]
and
\[
D\left(
A_j
\middle\Vert
\Theta_m+\frac1m A_j
\right)
=
D\left(
\rho_{\M\mathcal{H}}
\middle\Vert
\theta_{\M\mathcal{H}}+\frac1m\rho_{\M\mathcal{H}}
\right).
\]
Thus
\[
D(\Omega_m\Vert\Theta_m)
\leq
D(\rho_{\M\mathcal{H}}\Vert\theta_{\M\mathcal{H}})
-
D\left(
\rho_{\M\mathcal{H}}
\middle\Vert
\theta_{\M\mathcal{H}}+\frac1m\rho_{\M\mathcal{H}}
\right).
\]

Applying Theorem~\ref{general-convexsplit} with
\(
\lambda=\frac1m
\)
gives
\[
D(\Omega_m\Vert\Theta_m)
\leq
\frac{c_s \log(e)}{s}\,
2^{-s\left(
\log m-D_{1+s}(\rho_{\M\mathcal{H}}\Vert\theta_{\M\mathcal{H}})
\right)}.
\]
This proves the claim.
\end{proof}

\subsection{Proof of Theorem~\ref{thm:semifinite-decoupling-reliability}}
\label{subsec:semifinite-decoupling-proof}
From the generalized convex-split lemma, we can get the following estimate on quantum information decoupling:
\begin{proposition}
\label{prop:dimension-free-decoupling}
Under the assumptions of Proposition~\ref{prop:dimension-free-convex-split}, for
every \(m\in\mathbb N\), every \(\sigma_{\mathcal H}\in \mathcal{S}(\mathcal H)\), and
every \(0<s\le1\),
\[
    P^{\mathrm{dec-u}}_{\M:\mathcal H}
    (
        \rho_{\M\mathcal H},
        \log m
    )
    \le
    \sqrt{\frac{c_s(\log e)}{s}}\,
    2^{-\frac{s}{2}\left(
        \log m
        -
        D_{1+s}
        (
            \rho_{\M\mathcal H}
            \Vert
            \rho_{\M}\otimes\sigma_{\mathcal H}
        )
    \right)}.
\]
Equivalently, 
\[
    P^{\mathrm{dec}}_{\M:\mathcal H}
    \left(
        \rho_{\M\mathcal H},
        \frac12\log m
    \right)
    \le
    \sqrt{\frac{c_s(\log e)}{s}}\,
    2^{-\frac{s}{2}\left(
        \log m
        -
        D_{1+s}
        (
            \rho_{\M\mathcal H}
            \Vert
            \rho_{\M}\otimes\sigma_{\mathcal H}
        )
    \right)}.
\]
\end{proposition}
\begin{proof}
Let \(\Omega_m\) and \(\Theta_m\) be the two states in
Proposition~\ref{prop:dimension-free-convex-split}.  The random-unitary
decoupling protocol is the standard swap protocol: take \(m-1\) catalyst
copies of \(\sigma_{\mathcal H}\), choose uniformly one of the \(m\) registers,
and swap the original \(\mathcal H\)-system into that register.  The averaged
output state is precisely
\[
    \Omega_m
    =
    \frac1m
    \sum_{j=1}^m
    \rho_{\M\mathcal H_j}
    \otimes
    \bigotimes_{i\ne j}\sigma_{\mathcal H_i},
\]
while the target product state is
\[
    \Theta_m
    =
    \rho_{\M}\otimes\sigma_{\mathcal H}^{\otimes m}.
\]
Hence
\[
    P^{\mathrm{dec-u}}_{\M:\mathcal H}
    (
        \rho_{\M\mathcal H},
        \log m
    )
    \le
    P(\Omega_m,\Theta_m).
\]
For normalized states, the fidelity--relative entropy inequality gives
\[
    P(\Omega_m,\Theta_m)^2
    =
    1-F(\Omega_m,\Theta_m)^2
    \le
    D(\Omega_m\Vert\Theta_m).
\]
Combining this with Proposition~\ref{prop:dimension-free-convex-split}
yields
\[
    P(\Omega_m,\Theta_m)
    \le
    \sqrt{\frac{c_s(\log e)}{s}}\,
    2^{-\frac{s}{2}\left(
        \log m
        -
        D_{1+s}
        (
            \rho_{\M\mathcal H}
            \Vert
            \rho_{\M}\otimes\sigma_{\mathcal H}
        )
    \right)}.
\]
This proves the random-unitary estimate.  The subsystem-removal estimate
follows from the definition
\[
    P^{\mathrm{dec}}_{\M:\mathcal H}
    \left(
        \rho_{\M\mathcal H},
        \frac12\log m
    \right)
    =
    P^{\mathrm{dec-u}}_{\M:\mathcal H}
    (
        \rho_{\M\mathcal H},
        \log m
    ).
\]
\end{proof}

\begin{lemma}
\label{prop:semifinite-smooth-max-info}
Let $\M$ be a semifinite von Neumann algebra and $\mathcal{H}$ be a finite-dimensional Hilbert space. Let \[
    \rho_{\M\mathcal H}
    \in
    \mathcal{S}(
        \M\bar\otimes \mathcal{B}(\mathcal H),\tau_\M\otimes \operatorname{Tr}_{\mathcal{H}})
\]
be a normalized state. Define
\[
    \delta_{\M:\mathcal H}
    (
        \rho_{\M \mathcal H},\lambda
    )
    =
    \inf_{\sigma_{\mathcal H}\in \mathcal{S}(\mathcal H)}
    \Delta
    \bigl(
        \rho_{\M \mathcal H}
        \Vert
        \rho_{\M}\otimes\sigma_{\mathcal H},
        \lambda
    \bigr).
\]
Then, for every \(r\in\mathbb R, r\neq\frac{1}{2} I_\infty(\M:\mathcal{H})_\rho\),
\[
    \lim_{n\to\infty}
    -\frac1n
    \log
    \delta_{\M^{\bar\otimes n}:\mathcal H^{\otimes n}}
    \bigl(
        \rho_{\M\mathcal H}^{\otimes n},2nr
    \bigr)
    =
    \sup_{s\ge0}
    s\left(
        r-\frac12 I_{1+s}(\M:\mathcal H)_\rho
    \right).
\]
\end{lemma}

\begin{proof}
For this proof, we use the following notation. For a fixed $k_\M \in L^1({\mathcal{M}})_+$, define
$$ I_{\alpha}(\M:\mathcal{H})_{\rho,k}=\inf_{\sigmah\in{\mathcal{S}(\mathcal{H})}} D_\alpha(\rho_{\M\mathcal{H}} \Vert k_\M\otimes \sigmah),$$
and
\[
    \delta_{\M:\mathcal H}
    (
        \rho_{\M\mathcal H},\lambda
    )_k
    =
    \inf_{\sigma_{\mathcal H}\in S(\mathcal H)}
    \Delta
    \bigl(
        \rho_{\M\mathcal H}
        \Vert
        k_\M\otimes\sigma_{\mathcal H},
        \lambda
    \bigr)
\]
We also denote the related Legendre transform as
$$\sup_{s\ge0}
    s\left(
        r-\frac12 I_{1+s}(\M:\mathcal H)_{\rho,k}
    \right)=\sup_{s\ge0}\left\{sr-\varphi_{\rho,k}(s)
    \right\}=E_{\rho,k}(r). $$
It is sufficient to prove for the modified definition, and put $k_\M=\rho_\M$ one can recover the desired result.

The lower bound follows by restricting the optimizer in the smoothing
definition to product states
\[
    \sigma_{\mathcal H^{\otimes n}}
    =
    \sigma_{\mathcal H}^{\otimes n}
\]
and applying the result in Theorem~\ref{thm:semifinite-tracial-smoothing-exponent} to the pair \(
    (\rho_{\M\mathcal H}^{\otimes n}
    ,
    k_{\M}^{\otimes n}
    \otimes
    \sigma_{\mathcal H}^{\otimes n})\).
Optimizing over \(\sigma_{\mathcal H}\in \mathcal{S}(\mathcal H)\) gives the required ``$\geq$'' direction.

It remains to prove the converse direction. We follow a similar procedure in Theorem~\ref{thm:finite-tracial-positive-part-hoeffding}. We assume that $E_{\rho,k}(r)<\infty$, otherwise the conclusion is trivial.

\medskip
\noindent
\emph{Step 1: Finite spectrum case.}

Assume first that \({\mathcal{N}}\) is finite and
\(k_{\mathcal{N}}\) has finite spectrum. Then, for \(n\geq1\), let \(\mathcal E_n\) be the pinching map with respect to
\(k_{\mathcal{N}}^{\otimes n}\otimes \sigma_{\mathcal{H}^{\otimes n}}\).
Hence the proof of \cite[Theorem 15]{LiYao2024} holds for such case. Therefore, in this
finite-spectrum situation,
\[
\limsup_{n\to\infty}
-\frac1n
\log
\delta_{\mathcal{N}^{\bar\otimes n}:\mathcal{H}^{\otimes n}}
\bigl(\rho_{\mathcal{N}\mathcal{H}}^{\otimes n},2nr\bigr)_k
\leq
 E_{\rho,k}(r).
\]

\medskip
\noindent
\emph{Step 2: Finite corner result.}
Assume that $k_\NN\in L^1(\NN,\tau)_+$ is bounded and invertible. Now we can find a positive element $\pi_\varepsilon\in L^1(\mathcal{N})_+$ which has a finite-spectrum decomposition, such that
$$ 2^{-2\varepsilon}\pi_{\varepsilon}\leq k_{\mathcal{N}}\leq 2^{2\varepsilon}\pi_{\varepsilon}.$$

For all candidates $\psi\leq 2^{2nr}k_\mathcal{N}^{\otimes n}\otimes \sigma_{\mathcal{H}^{\otimes n}}$, we have
$$ {\psi}\leq 2^{2n(r+\varepsilon)}\pi_\varepsilon^{\otimes n}\otimes \sigma_{\mathcal{H}^{\otimes n}}.$$
Hence we get
$$ \delta_{\mathcal{N}^{\bar\otimes n}:\mathcal{H}^{\otimes n}}
\bigl(\rhonh^{\otimes n},2n(r+\varepsilon)\bigr)_{\pi_\varepsilon}\leq \delta_{\mathcal{N}^{\bar\otimes n}:\mathcal{H}^{\otimes n}}
\bigl(\rhonh^{\otimes n},2nr\bigr)_k$$
Similarly,
$$ \delta_{\mathcal{N}^{\bar\otimes n}:\mathcal{H}^{\otimes n}}
\bigl(\rhonh^{\otimes n},2nr\bigr)_k\leq \delta_{\mathcal{N}^{\bar\otimes n}:\mathcal{H}^{\otimes n}}
\bigl(\rhonh^{\otimes n},2n(r-\varepsilon)\bigr)_{\pi_\varepsilon}$$
Applying the finite-spectrum results on $\pi_\varepsilon$ and letting $n\to \infty$, we get
\begin{align}\label{sand-estimate}
E_{\rho,\pi_\varepsilon}(r-\varepsilon)\leq\limsup_{n\to\infty}
-\frac1n
\log
\delta_{\mathcal{N}^{\bar\otimes n}:\mathcal{H}^{\otimes n}}
\bigl(\rho_{\mathcal{N}\mathcal{H}}^{\otimes n},2nr\bigr)_k
\leq
 E_{\rho,\pi_\varepsilon}(r+\varepsilon).
\end{align}
Since $E_{\rho,k}(r)$ is the supremum function of a family of affine functions with respect to $r$, $E_{\rho,k}(r)$ is convex. Hence it is continuous at $r$.

Besides, the choice of $\pi_\varepsilon$ implies the uniform estimate with respect to $\sigma_\mathcal{H}$
$$ \left| D_{1+s}(\rho_{\mathcal{NH}}\Vert \pi_\varepsilon \otimes \sigma_\mathcal{H})-D_{1+s}(\rho_{\mathcal{NH}}\Vert k_\mathcal{N} \otimes \sigma_\mathcal{H})\right|\leq 2\varepsilon.$$
Hence the estimate can be passed to $I_{1+s}(\mathcal{N}:\mathcal{H})_k$. Let $\varepsilon \to 0$, by Lemma~\ref{l-trans} we get the following uniform convergence
$$\sup_{s\geq 0} s\left(r-\frac{1}{2}D_{1+s}(\rho_{\mathcal{NH}}\Vert \pi_\varepsilon \otimes \sigma_\mathcal{H})\right)\to \sup_{s\geq 0} s\left(r-\frac{1}{2}D_{1+s}(\rho_{\mathcal{NH}}\Vert k_\mathcal{N} \otimes \sigma_\mathcal{H})\right).$$
Taking the supremum over $\sigma_\mathcal{H}\in \mathcal{S}(\mathcal{H})$ we obtain
 $E_{\rho,\pi_\varepsilon}(r)\to E_{\rho,k}(r)$. Combining \eqref{sand-estimate} we get
 $$ \limsup_{n\to\infty}
-\frac1n
\log
\delta_{\mathcal{N}^{\bar\otimes n}:\mathcal{H}^{\otimes n}}
\bigl(\rho_{\mathcal{N}\mathcal{H}}^{\otimes n},2nr\bigr)_k
\leq
 E_{\rho,k}(r).$$
This completes the finite-corner case.

\medskip
\noindent
\emph{Step 3: Semifinite recovery.}
Denote $e_{\delta,N}=\1_{[\delta,N]}(k_{\M})$ for $0<\delta<N<\infty$. Hence $\mathcal{N}=e_{\delta,N} \M e_{\delta,N}\oplus \mathbb{C}$ is a finite algebra. Put
\begin{align*}
  \Phi:  L^1(\M,\tau_\M)& \to  L^1(\NN,\tau_{\NN}),\\
   x\;\;\;\;&\mapsto \Big(e_{\delta,N}xe_{\delta,N},\tau_\M\big((\1_\M-e_{\delta,N})x\big)\Big).
\end{align*}
This is a CPTP map with $\widetilde{\Phi}=\Phi \otimes \operatorname{id}_\mathcal{H}$ its extension. The data-processing inequality of the purified distance gives
$$\Delta(\widetilde{\Phi}(\rho_{\mathcal{M}\mathcal{H}})\Vert \widetilde{\Phi}(k_{\mathcal{M}}\otimes \sigma_{\mathcal{H}}),2nr)\leq \Delta(\rho_{\mathcal{M}\mathcal{H}}\Vert k_{\mathcal{M}}\otimes \sigma_{\mathcal{H}},2nr),$$
from which we can deduce
\begin{align*}\limsup_{n\to\infty}
\frac{-1}n
\log
\delta_{\M^{\bar\otimes n}:\mathcal{H}^{\otimes n}}
&\bigl(\rho_{\M \mathcal{H}}^{\otimes n},2nr\bigr)_{k}\\
&\leq\limsup_{n\to\infty}
\frac{-1}n
\log
\delta_{\mathcal{N}^{\bar\otimes n}:\mathcal{H}^{\otimes n}}
\bigl(\widetilde{\Phi}(\rho_{\M \mathcal{H}})^{\otimes n},2nr\bigr)_{\Phi(k)}\\
&\leq E_{\Phi(\rho),\Phi(k)}(r).
\end{align*}
Here the last inequality follows from \emph{Step 2}.

Moreover, with a similar uniform convergence discussion in Step 2, the only task now is to check
$$ D_{1+s}(\widetilde{\Phi}(\rho_{\M\mathcal{H}})\Vert \widetilde{\Phi}(k_\M\otimes \sigma_\mathcal{H}))\to D_{1+s}(\rho_{\M\mathcal{H}}\Vert k_\M\otimes \sigma_\mathcal{H}) $$
with $\delta \to 0 $ and $N \to \infty.$
For this purpose, we write
\begin{align*}
     &D_{1+s}(\rho_{\M\mathcal{H}}\Vert k_\M \otimes \sigma_\mathcal{H})\\
    \geq&D_{1+s}(\widetilde{\Phi}(\rho_{\M\mathcal{H}})\Vert \widetilde{\Phi}(k_\M\otimes \sigma_\mathcal{H}))\\
    \geq& D_{1+s}(e \rho_{\M\mathcal{H}}e \Vert e(k_\M \otimes \sigma_\mathcal{H})e)
\end{align*}
where $e=e_{\delta,N}\otimes \1_\mathcal{H}$. The first inequality follows from the data processing inequality, and the second follows from discarding the second part in the direct sum. Then Proposition~\ref{thm:finite-trace-recoverability} implies the convergence.

Finally, setting $k_\M=\rho_\M$, we prove the desired result.
\end{proof}

Applying the Haagerup--Junge--Xu reduction as in Theorem~\ref{thm:general-smoothing-exponent}, we get the following result.
\begin{theorem}
\label{general-smooth-max-info}
Let $\M$ be a von Neumann algebra and $\mathcal{H}$ be a finite-dimensional Hilbert space. Let \[
    \rho_{\M\mathcal H}
    \in
    \mathcal{S}(
        \M\bar\otimes \mathcal{B}(\mathcal H))
\]
be a normalized state. Define
\[
    \delta_{\M:\mathcal H}
    (
        \rho_{\M \mathcal H},\lambda
    )
    =
    \inf_{\sigma_{\mathcal H}\in \mathcal{S}(\mathcal H)}
    \Delta
    \bigl(
        \rho_{\M \mathcal H}
        \Vert
        \rho_{\M}\otimes\sigma_{\mathcal H},
        \lambda
    \bigr).
\]
Then, for every \(r\in\mathbb R, r\neq\frac{1}{2} I_\infty(\M:\mathcal{H})_\rho\),
\[
    \lim_{n\to\infty}
    -\frac1n
    \log
    \delta_{\M^{\bar\otimes n}:\mathcal H^{\otimes n}}
    \bigl(
        \rho_{\M\mathcal H}^{\otimes n},2nr
    \bigr)
    =
    \sup_{s\ge0}
    s\left(
        r-\frac12 I_{1+s}(\M:\mathcal H)_\rho
    \right).
\]
\end{theorem}

\begin{proof}
Restricting the optimization to product states
\(\sigma_{\mathcal H}^{\otimes n}\) and applying
Theorem~\ref{thm:general-smoothing-exponent} gives
\[
\begin{aligned}
\liminf_{n\to\infty}
-\frac1n\log
\delta_{\mathcal M^{\bar\otimes n}:\mathcal H^{\otimes n}}
\bigl(\rho_{\mathcal M\mathcal H}^{\otimes n},2nr\bigr)
\geq
\sup_{s\geq0}
s\left(
r-\frac12 I_{1+s}(\mathcal M:\mathcal H)_\rho
\right).
\end{aligned}
\]

For the reverse inequality, let
\(\rho_{\mathcal R_m\mathcal H}\) be the restrictions of the canonical
lift of \(\rho_{\mathcal M\mathcal H}\) to
\(\mathcal R_m\bar\otimes\mathcal B(\mathcal H)\).
After data processing to $\RR_m \bar\otimes\mathcal B(\mathcal H)$, Lemma~\ref{prop:semifinite-smooth-max-info} implies
\[
\begin{aligned}
\limsup_{n\to\infty}
-\frac1n\log
\delta_{\mathcal M^{\bar\otimes n}:\mathcal H^{\otimes n}}
\bigl(\rho_{\mathcal M\mathcal H}^{\otimes n},2nr\bigr)
\leq
\sup_{s\geq0}
s\left(
r-\frac12 I_{1+s}(\RR_m:\mathcal{H})_{\rho_{\mathcal R_m\mathcal H}}
\right).
\end{aligned}
\]
Combining the martingale recovery of the $D_\alpha$ in Proposition~\ref{prop:haagerup-information-approximation}, the compactness of
\(\mathcal S(\mathcal H)\) and Lemma~\ref{l-trans}, we obtain
\[
   \sup_{s\geq0}
s\left(
r-\frac12 I_{1+s}(\RR_m:\mathcal{H})_{\rho_{\mathcal R_m\mathcal H}}
\right)
    \longrightarrow
    \sup_{s\geq0}
s\left(
r-\frac12 I_{1+s}(\M:\mathcal{H})_{\rho}
\right) .
\]
This proves the claim.
\end{proof}

\begin{proposition}
\label{prop:semifinite-decoupling-converse}
Let $\M$ be a von Neumann algebra. For every state $\rho_{\M\mathcal H}\in \mathcal{S}(\M\bar{\otimes}\mathcal{B}(\mathcal{H}))$ and every \(k\ge0\),
\[
    P^{\mathrm{dec}}_{\M:\mathcal H}
    (
        \rho_{\M\mathcal H},k
    )
    \ge
    \delta_{\M:\mathcal H}
    (
        \rho_{\M\mathcal H},2k
    ).
\]
\end{proposition}

\begin{proof}
The result follows from \cite[Proposition~19]{LiYao2024}, since the proof only relies on the operations on the finite-dimensional part $\mathcal{B}(\mathcal{H})$. We mention that the Uhlmann's theorem used there also holds for states on arbitrary von Neumann algebras, e.g., see \cite{Uhlmann1976}, not only for matrix algebras.
\end{proof}

\begin{proof}[Proof of Theorem~\ref{thm:semifinite-decoupling-reliability}]
Let $\M$ be a von Neumann algebra. For the ``$\leq$'' part,   Proposition~\ref{prop:semifinite-decoupling-converse}
gives
\[
    P^{\mathrm{dec}}_{\M^{\bar\otimes n}:\mathcal H^{\otimes n}}
    \bigl(
        \rho_{\M \mathcal H}^{\otimes n},nr
    \bigr)
    \ge
    \delta_{\M^{\bar\otimes n}:\mathcal H^{\otimes n}}
    \bigl(
        \rho_{\M\mathcal H}^{\otimes n},2nr
    \bigr).
\]
Using Theorem~\ref{general-smooth-max-info}, we obtain
\[
\begin{aligned}
    E^{\mathrm{dec}}_{\M :\mathcal H}
    (
        \rho_{\M \mathcal H},r
    )
    &\le
    \lim_{n\to\infty}
    -\frac1n
    \log
    \delta_{\M^{\bar\otimes n}:\mathcal H^{\otimes n}}
    \bigl(
        \rho_{\M\mathcal H}^{\otimes n},2nr
    \bigr)                                                  \\
    &=
    \sup_{s\ge0}
    s\left(
        r-\frac12 I_{1+s}(\M:\mathcal H)_\rho
    \right).
\end{aligned}
\]

For the ``$\geq$'' part.  Proposition~\ref{prop:dimension-free-decoupling} shows that for every state $\rho_{\M\mathcal{H}}$, every $\sigma_{\mathcal{H}} \in \mathcal{S}(\mathcal{H})$, every
$0<s\le 1$, and every $k\ge 0$,
\[
   P^{\mathrm{dec}}_{\M:H}(\rho_{\M\mathcal{H}},k)
   \le C_s\,
   2^{-s\left(k-\frac12
   D_{1+s}(\rho_{\M\mathcal{H}}\Vert \rho_{\M}\otimes\sigma_{\mathcal{H}})\right)},
\]
where $C_s<\infty$ depends only on $s$. Apply this one-shot bound to
\[
   \rho_{\M\mathcal{H}}^{\otimes n}
   \quad\text{and}\quad
   \sigma_{\mathcal{H}}^{\otimes n},
\]
with cost $k=nr$.  By additivity of the sandwiched Rényi divergence,
\[
   D_{1+s}\!\left(
        \rho_{\M\mathcal{H}}^{\otimes n}
        \Big\Vert
        \rho_{\M}^{\otimes n}\otimes \sigma_{\mathcal{H}}^{\otimes n}
   \right)
   =
   nD_{1+s}(\rho_{\M\mathcal{H}}\Vert \rho_{\M}\otimes\sigma_{\mathcal{H}}).
\]
Hence
\[
   P^{\mathrm{dec}}_{\M^{\bar\otimes n}:\mathcal{H}^{\otimes n}}
   \left(\rho_{\M\mathcal{H}}^{\otimes n},nr\right)
   \le
   C_s\,
   2^{-ns\left(
        r-\frac12D_{1+s}(\rho_{\M\mathcal{H}}\Vert \rho_{\M}\otimes\sigma_{\mathcal{H}})
   \right)} .
\]
Taking logarithms, then taking the limit,  we obtain
\[
   \liminf_{n\to\infty}
   -\frac1n
   \log
   P^{\mathrm{dec}}_{\M^{\bar\otimes n}:\mathcal{H}^{\otimes n}}
   \left(\rho_{\M\mathcal{H}}^{\otimes n},nr\right)
   \ge
   s\left(
        r-\frac12D_{1+s}(\rho_{\M\mathcal{H}}\Vert \rho_{\M}\otimes\sigma_{\mathcal{H}})
   \right).
\]
Since this holds for every $\sigma_{\mathcal{H}}\in \mathcal{S}(\mathcal{H})$, we may optimize over
$\sigma_{\mathcal{H}}$ and obtain
\[
   E^{\mathrm{dec}}_{\M:\mathcal{H}}(\rho_{\M\mathcal{H}},r)
   \ge
   s\left(
        r-\frac12 I_{1+s}(\M:\mathcal{H})_\rho
   \right).
\]
Finally, optimizing over $0<s\le 1$, and including the trivial value
$s=0$, yields
\[
   E^{\mathrm{dec}}_{\M:\mathcal{H}}(\rho_{\M\mathcal{H}},r)
   \ge
   \sup_{0\le s\le 1}
   s\left(
        r-\frac12 I_{1+s}(\M:\mathcal{H})_\rho
   \right).
\]

It remains to identify the region where the two bounds coincide.  Let
\[
    f(s)
    =
    s\left(
        r-\frac12 I_{1+s}(\M:\mathcal H)_\rho
    \right),
    \qquad s\ge0.
\]
The map
\[
    s\mapsto s I_{1+s}(\M:\mathcal H)_\rho
\]
is convex, and therefore \(f\) is concave.  If
\[
    f'(1)\le0,
\]
then the maximum of \(f\) over \([0,\infty)\) is already attained on
\([0,1]\).  This condition is precisely
\[
    r\le
    \frac12
    \left.
    \frac{d}{ds}
    \bigl[
        s I_{1+s}(\M:\mathcal H)_\rho
    \bigr]
    \right|_{s=1}
    =
    R^\sharp.
\]
Hence, for \(r\le R^\sharp\),
\[
    \sup_{s\ge0} f(s)
    =
    \max_{0\le s\le1} f(s),
\]
and the achievability and converse bounds match.
\end{proof}

\appendix
\section{Legendre transform lemma}\label{app-ltrans}
\newtheorem*{lemma*}{Lemma \ref{l-trans}}
\begin{lemma*}
Let \(I\) be a directed set. Let
\[
0\neq\rho_i,\sigma_i\in(\mathcal M_i)_*^+,
\qquad
0\neq\rho,\sigma\in\mathcal M_*^+.
\]
Put
\(
\rho_i(\1)\to \rho(\1),\,\sigma_i(\1)\to \sigma(\1)
\).
For \(s>0\), put
\[
\varphi_i(s)
=
sD_{1+s}(\rho_i\Vert\sigma_i),
\quad
\varphi(s)
=
sD_{1+s}(\rho\Vert\sigma).
\] 
with natural extension $\varphi_i(0)=\log \rho_i(\1)$, $\varphi(0)=\log \rho(\1)$.

If \(
a<D_\infty(\rho\Vert\sigma)
\) and $\varphi_i(s)\to \varphi(s)$ pointwise, then
\[
J_i(a)=\inf_{s\geq0}
\{\varphi_i(s)-a s\}
\longrightarrow
\inf_{s\geq0}
\{\varphi(s)-as\}=J(a).
\]
Equivalently,
\[
E_i(a)=\sup_{s\geq0}
\{a s-\varphi_i(s)\}
\longrightarrow
\sup_{s\geq0}
\{as-\varphi(s)\}=E(a).
\]
\end{lemma*}

\begin{proof}
Without loss of generality,  we can reduce to the case where
\(\rho_i,\rho\) are normalized states. Indeed, one has
\[
\varphi_i(s)
=
\log\rho_i(\1)
+\overline\varphi_i(s),\quad \varphi(s)
=
\log\rho(\1)
+\overline\varphi(s)
\]
where \(\overline\varphi_i\) is defined for $\frac{{\rho}_i}{\rho_i(\1)},\frac{\sigma_i}{\rho_i(\1)}$ and \(\overline\varphi\) is defined similarly. Hence the additive constant $\rho_i(\1)\to \rho{(\1)}$ will not affect the result.

Then, one can  easily check, by the properties of the relative entropy, $\varphi_i$ is convex and
\(\varphi_i(0)=0\). Hence the function
\[
s\longmapsto \frac{\varphi_i(s)}{s}
\]
is nondecreasing on \((0,\infty)\). The same holds for \(\varphi\).

Choose \(S_0>0\) and \(\eta>0\) such that
\[
\frac{\varphi(S_0)}{S_0}>a+2\eta.
\]
The value \(\varphi(S_0)=+\infty\) is allowed. By pointwise convergence, for all sufficiently large \(i\),
\[
\frac{\varphi_i(S_0)}{S_0}>a+\eta.
\]
Consequently, for every \(s\geq S_0\),
\[
\varphi_i(s)-as
=
s\left(\frac{\varphi_i(s)}{s}-a\right)
\geq \eta s>0.
\]
On the other hand,
\[
 J_i(a)\leq\varphi_i(0)=0.
\]
Thus, for all sufficiently large \(i\), the infimum of
\( J_i(a)\) may be restricted to the compact interval
\([0,S_0]\).

We shall use the following standard consequence of convexity: if
\(s_i\to s\), then
\[
\liminf_i\varphi_i(s_i)\geq\varphi(s).
\]
Choose \(s_i\in[0,S_0]\) such that
\[
\varphi_i(s_i)-a s_i
\leq
J_i(a)+o(1).
\]
After passing to a subnet, we may assume that \(s_i\to\widetilde{s}\in[0,S_0]\).
Hence
\[
\begin{aligned}
\liminf_i J_i(a)\geq
\liminf_i\{\varphi_i(s_i)-a s_i\}
\geq
\varphi(\widetilde{s})-a\widetilde{s}\geq J(a).
\end{aligned}
\]

For the reverse inequality, let \(\varepsilon>0\) and choose
\(s_\varepsilon\geq0\) such that
\[
\varphi(s_\varepsilon)-as_\varepsilon
\leq
 J(a)+\varepsilon.
\]
By pointwise convergence,
\[
\begin{aligned}
\limsup_iJ_i(a)\leq
\lim_i
\{\varphi_i(s_\varepsilon)-a s_\varepsilon\}=
\varphi(s_\varepsilon)-as_\varepsilon\leq J(a)+\varepsilon.
\end{aligned}
\]
Letting \(\varepsilon\downarrow0\) proves
\(
 J_i(a)\longrightarrow J(a)
\). The assertion for \( E_i\) follows from
\[
 E_i(a)=- J_i(a),
\qquad
 E(a)=- J(a).
\]
This completes the proof.
\end{proof}

\backmatter
\bmhead{Acknowledgements}
Hongsen Qiu and Xinyu Zhang are supported by the National
Natural Science Foundation of China (Grant No. 12371138 and No. W2441002). Zhiwen Lin is supported by the National
Natural Science Foundation of China (Grant No. 62571166, No. 12371138 and No. W2441002). They would like to thank their supervisor, Prof. Xiao Xiong, for his guidance and support throughout the research. They would also like to thank Prof. Li Gao for his valuable comments and suggestions.

\bmhead{Data Availability} Data sharing is not applicable to this article, since no data sets were used.

\bmhead{Conflict of Interest} The authors declare that they have no conflict of interest.

\bibliography{references_ams}

\end{document}